\documentclass[journal]{IEEEtran}

\usepackage{graphicx}
\usepackage[colorlinks,bookmarksopen,bookmarksnumbered,citecolor=blue,urlcolor=red]{hyperref}
\usepackage{amssymb}
\usepackage{amsbsy}
\usepackage{ulem}
\usepackage{upgreek}
\usepackage{algorithm}
\usepackage{algorithmicx}
\usepackage{algpseudocode}
\usepackage{threeparttable}
\usepackage[cmex10]{amsmath}
\usepackage{array}
\usepackage[tight,footnotesize]{subfigure}
\usepackage{caption}
\usepackage{url}
\usepackage{epstopdf}
\usepackage{verbatim}
\usepackage{color}

\usepackage{xspace}
\newcommand*{\eg}{e.g.\@\xspace}

\makeatletter
\newcommand*{\etc}{%
    \@ifnextchar{.}%
        {etc}%
        {etc.\@\xspace}%
}
\graphicspath{{figures/}}

\begin{document}

\newtheorem{thm}{Theorem}
\newtheorem{lma}{Lemma}
\newtheorem{defi}{Definition}
\newtheorem{proper}{Property}

\title{Challenges, Designs, and Performances of a Distributed Algorithm for Minimum-Latency of Data-Aggregation in Multi-Channel WSNs}

\author{Ngoc-Tu~Nguyen,~
    Bing-Hong~Liu,
    Shao-I Chu,~and
    Hao-Zhe~Weng
\IEEEcompsocitemizethanks{\IEEEcompsocthanksitem N.T.~Nguyen is
with the Department of Computer Science \& Engineering, University of Minnesota, Twin Cities, Minneapolis, MN 55455, USA.
\IEEEcompsocthanksitem E-mail: nguy3503@umn.edu.}
\IEEEcompsocitemizethanks{\IEEEcompsocthanksitem B.H.~Liu, Shao-I Chu, and H.Z~Weng are
with the Department of Electronic Engineering, National Kaohsiung
University of Applied Sciences, 415, Chien Kung Rd., Kaohsiung
80778, Taiwan.
\IEEEcompsocthanksitem E-mail: bhliu@kuas.edu.tw, erwinchu@kuas.edu.tw, 1103305142@gm.kuas.edu.tw.}
}

\IEEEcompsoctitleabstractindextext{%
\begin{abstract}
In wireless sensor networks (WSNs), the sensed data by sensors
need to be gathered, so that one very important application is periodical data collection.
There is much effort which aimed at the data collection scheduling algorithm development
to minimize the latency. Most of previous works investigating the minimum latency of data collection
issue have an ideal assumption that the network is a \textit{centralized system},
in which the entire network is completely synchronized with full knowledge
of components. In addition, most of existing works often assume that
any (or no) data in the network are allowed to be aggregated into one packet and
the network models are often treated as tree structures.
However, in practical, WSNs are more likely to be \textit{distributed systems},
since each sensor's knowledge is disjointed to each other, and
a fixed number of data are allowed to to be aggregated into one packet.
This is a formidable motivation for us to investigate the problem of minimum latency for the data aggregation without data
collision in the distributed WSNs when the sensors are considered to be assigned the channels and
the data are compressed with a flexible aggregation ratio, termed
the minimum-latency collision-avoidance multiple-data-aggregation scheduling with multi-channel (MLCAMDAS-MC) problem.
A new distributed algorithm, termed the distributed collision-avoidance scheduling (DCAS) algorithm,
is proposed to address the MLCAMDAS-MC. Finally, we provide the theoretical analyses of DCAS
and conduct extensive simulations to demonstrate the performance of DCAS.
\end{abstract}
}

\maketitle

\IEEEdisplaynotcompsoctitleabstractindextext

\IEEEpeerreviewmaketitle

\section{Introduction}
In Wireless Sensor Networks (WSNs), one very important application is periodical data collection.
Generally, data collection can be viewed under two stages as \textit{data generation}, in which the data
are periodically generated by sensors, and \textit{data aggregation}, in which after the generated data from the sensors
compressed by using some aggregation functions, \eg MAX, MIN, SUM, and \etc are reported to a specific node, called the sink.
Since the \textit{data generation} was fulfilled because of the development of sensing capacity of the sensors,
the prerequisite for deciding the performance of WSNs is \textit{data aggregation} capacity that
reflects how fast data been collected by the sink.

To enable efficient data collection in the WSNs, for the last $10$ years
there are a lots of effort which aimed at the routing algorithm development with collision free when a fixed number of data
are allowed to be aggregated into one packet \cite{Madden:2002:TTA:844128.844142}.
It is worth mentioning that there exists a problem called minimum-latency collision-avoidance multiple-data-aggregation scheduling (MLCAMDAS)
\cite{8051077} that is successfully investigated minimum the latency of data collection with a flexible aggregation ratio $\alpha$ in WSNs,
in which $\alpha$ is the maximum number of data allowed to be aggregated into one packet.
However, in reality, to guarantee the data collision free sensors are not only assigned to time slots reasonably but also
need to be set communication channels. It is essential to guarantee the collision free and there is no above research works fully consider
a communication channels scheduling for the data collection application in the WSNs.
In addition, most of existing works studied the data collection issue in the WSNs
under an ideal assumption that the entire network is completely synchronized with full knowledge
of components. It is commonly known as \textit{centralized wireless sensor network} or generally \textit{centralized system}.
Likewise, many centralized algorithms based on the above assumption are designed and analyzed with nice performance
to solve problems of data collection in the \textit{centralized wireless sensor network}.
In partial application, the WSNs are more likely to be \textit{distributed systems},
in which each sensor's knowledge is disjointed to each other, even the sink
is not with full information of the network either.
This is a formidable motivation for us to investigate the problem of
minimum latency for the data aggregation application without data
collision in the distributed WSNs when the sensors are considered to be assigned the channels and
the data are compressed with a flexible aggregation ratio, termed
the minimum-latency collision-avoidance multiple-data-aggregation scheduling with multi-channel (MLCAMDAS-MC) problem.

During studying the MLCAMDAS-MC problem, many new challenges are realized and compared with that in
previous works. We summarize these main challenges as follows.

\begin{itemize}
  \item \textbf{C1:} To guarantee the collision free, every sensor node are not only required to be assigned to time slots, but also
  need to be set to appropriate channels. Unlike in the centralized system,
  they can collect overall information of the network and give the time slot
  and the channel assignment evaluations based on all sensor nodes' information
  such as the time clock, the packet size, and the links between the sensors \etc
  to schedule collision free algorithms for data transmissions in the network,
  in a distributed WSN, we have to guarantee the collision free based on only local
  information of each sensor node. It is clearly much harder and requires a complex technique
  to provide the collision free algorithms for data collection in the distributed WSNs.

  \item \textbf{C2:} To achieve the minimum latency for data collection, it requires
  an efficient distributed algorithm in the distributed WSN. In the centralized wireless sensor network,
  they can provide an optimized algorithm for data transmissions in the network because
  the entire network is synchronized completely with full knowledge information of all components.
  As always, following those existing optimized algorithms are no longer suitable for a distributed WSN.
  Thus, how to design an optimized distributed algorithm for data collection in the distributed WSNs
  is a challenge.

  \item \textbf{C3:} The third challenge is how to theoretically analyze the framework to
  design an actual distributed algorithm. Since there is no way to get exactly the
  parameters of multiple sensors at the same time, it is difficult to determine
  relevant routings for the data transmissions of the sensor nodes. The question of guarantee for
  minimum data collection's latency and collision free becomes harder and harder.
  Hence, the performance of the desired algorithm not only depends on
  an comprehensive evaluation but also requires the specific techniques and mechanism for the data transmissions.
\end{itemize}

To address these challenges, we proposed a new distributed algorithm,
termed the distributed collision-avoidance scheduling (DCAS) algorithm
to let every sensor $u$ iteratively schedule local data transmissions
for each time slot by its 3-hop neighboring information, until no more data are required to be scheduled by $u$.
On the comprehensiveness, we implement our proposed method through simulations and analyses.
We summarized the main contributions of this paper as follows.

\begin{itemize}
  \item We study the problem of finding a schedule of forwarding data to
    the sink without data collisions such that the number of required
    time slots is minimized, termed the MLCAMDAS-MC problem. In addition, the difficulty of the MLCAMDAS-MC problem is provided.
  \item We introduce an extended relative collision graph $G_{r+}$ to represent the
    collision relation between any data transmission for the MLCAMDAS-MC problem.
    Based on the obtained $G_{r+}$, we propose a new distributed algorithm,
    termed the distributed collision-avoidance scheduling (DCAS) algorithm.
  \item Theoretical analyses of the DCAS show its correctness. It indicates that the data transmissions scheduled by the DCAS
    are collision-free in each time slot as well as no sensors are in a circular wait for making schedules of data transmissions in the DCAS.
    We also conduct the simulations to demonstrate the performance of the DCAS. The results show that the DCAS provides
    better performance than the existing solution used for the MLCAMDAS-MC problem with an adjustable data aggregation ratio $\alpha$.
\end{itemize}

\textbf{Organization:}
The remaining sections of this paper are organized as follows.
A summary of the related works is in Section \ref{relatedw} to give the readers a whole picture
on the data collection and the distributed WSN. In Section \ref{SysMod}, the network model is introduced.
The MLCAMDAS-MC problem and its difficulty are illustrated in Section \ref{scheduling}.
We introduce an extended relative collision graph $G_{r+}$ in Section \ref{relative}.
According to the $G_{r+}$, the DCAS algorithm is presented in Section \ref{sec:proposed_algorithm}.
The theoretical analyses are provided in Section \ref{section:analysis}.
In Section \ref{section:Simulation}, we evaluate the performance of the DCAS.
Finally, this paper is concluded in Section \ref{section:Conclusion}.

\section{Related Work}\label{relatedw}
Many research works have studied to improve the efficiency of data collection
for both the centralized WSNs \cite{6570720,8051077} and distributed WSNs \cite{5462033,6195594}.
In \cite{1035242}, the authors proposed a chain-based protocol, named PEGASIS
to reduce the energy consumption of sensors during data collection process.
The idea of the PEGASIS is to collect data through a connected chain through sensors to the sink.
The authors in \cite{4195130} improve the PEGASIS by grouping the sensors into clusters.
The data are forwarded to the sink from the sensors through the cluster head.
However, in these studies they do not consider to eliminate the data collisions of data transmissions.
In \cite{6195594}, a distributed data collection algorithm is proposed
to increase the data collection capacity, in which the network is organized as a connected dominating set and the data can be collected through
the dominators. Even the authors in \cite{6195594} claim to be distributed algorithm;
unfortunately, during the construction of the Connected Dominating Set for
the network, they have accidentally treated the network as an centralize system,
so that it cannot be applied for an completely distributed network.

In WSNs, sensor may incur the data collision with the others during the data transmission,
resulting in data loss or failure. In recent years, there has been much effort to to
improve the latency of data collection as well as eliminate the data collision. The
time-division multiple access (TDMA) is one of the most common channel access techniques
used medium access control (MAC) protocol to allow multiple sensors to transmit data without collisions
at the same time slot. In \cite{5935125}, the authors proposed a novel tree-based TDMA scheduling, named
the traffic pattern oblivious (TPO) to achieve data collision free for data collection.
In this work, all sensor nodes in the network are step by step assigned to time slots
to avoid the data collisions. Since the network is structured as a tree structure,
the time slot assignment is conducted from the leaf nodes forward.
The authors in \cite{Ergen2010} proposed the node-based scheduling algorithm (NBSA)
and the level-based scheduling algorithm (LBSA) to minimize the latency of
data collection. In this work, the NBSA and the LBSA use a color graph to represent
the data collision between sensors' data transmission. Essentially, the above studies
achieve the data collision free and data collection capacity. However,
the data routings are limited in using the tree-based methods to minimize the latency of data collection.
In addition, the raw data are not aggregated before transmitting through the network.

Since reducing the sensors' energy consumption is a big challenge of
data collection application in WSNs, data aggregation technology appears as the best solution to allow data to be aggregated
in the data collection. Data aggregation uses the functions of
MIN, MAX, SUM and COUNT \etc to aggregate multiple packets into one packet.
In \cite{5462033,6195659}, the authors investigated the construction of data aggregation
with minimum energy cost in WSNs, in which the data aggregation technique
is applied to reduce the size of data packets. The authors in \cite{5444884,7046421}
proposed methods based on the connected-dominating-set (CDS) tree
to minimize the latency and achieve collision free in WSNs.
In the studies, the main idea is to organize the network as
a CDS tree, in which each sensor is treated as either a dominator or
a dominatee. The dominators are responsible to aggregate all data from
dominatees and the data collisions are eliminated by using the TDMA.
In these research works, the authors studied the data collection issue
under an ideal assumption that the network is structured with the tree root be
the topology center. The communication capacity of sensors in WSNs is recently enhanced by applying
the multi-channel technology, in which sensors can use the IEEE 802.15.4 protocol with
16 non-interference channels \cite{6747982,Bagaa2014293,SOUA20152}. By this way, the sensors can use different channels to
transmit data without collisions in the same time slot. Hence, the number of used time slots in
the data collection can be reduced, resulting in the improvement of the latency.
In \cite{6570720}, the authors introduced an idea to using maximum the number of channels to schedule the data transmission of sensors
in one time slot; however, since the number of channels is limited,
it is possible of occurring data collisions. Eventually, it requires
a perfect scheduling to combine the time slot assignment and the channel assignment for sensors
to achieve the minimum latency of data collection in WSNs.

\section{Preliminaries}\label{ModelProb}
In this section, we first describe the network model for a WSN in
Section \ref{SysMod}. Based on the network model, the
Minimum-Latency Collision-Avoidance Multiple-Data-Aggregation
Scheduling with Multi-Channel (MLCAMDAS-MC) problem and its
difficulty are proposed and discussed in Section \ref{scheduling}.

\subsection{Network Model} \label{SysMod}
The WSN is composed of sensors, where a sensor can communicate with
other sensor if and only if they are within each other's
transmission range. Hereafter, a sensor $u$ is said to be a sensor $v$'s neighboring sensor
if and only if sensors $u$ and $v$ can communicate with each other.
In this paper, the unit disk graph model is
employed as the communication model \cite{clark1990unit}, in which
all sensors are assumed to have the same transmission range, denoted
by $R_t$. Because sensors are responsible for periodically sensing
environmental information, sensing data are periodically generated
from sensors and reported to a sink, where the sink is a special
node in the network and is responsible for data collecting,
processing, and analysis. The wireless sensor network can then be
represented as a connected weighted graph $G$ $=$ $(V_G,E_G,\rho_G)$
\cite{8051077}, where $V_G$ is the set of sensors in the network, edge $(u,v)
\in E_G$ represents that sensors $u$ and $v$ can communicate with
each other, $\rho_G(v)$, the weighting function of sensor $v \in
V_G$, represents that the number of units of raw data generated by
$v$ to report to the sink $s \in V_G$ per a period of time. Fig.
\ref{Fig:net_model_org} shows a WSN represented by a connected
weighted graph $G = (V_G,E_G,\rho_G)$, which includes one sink
$s$ and $7$ sensors $a$, $b$, $c$, $d$, $e$, $f$, and $g$. The
numbers of units of raw data generated by $a$, $b$, $c$, $d$, $e$,
$f$, and $g$ per a period of time are $7$, $4$, $4$,
$2$, $1$, $1$, and $1$, respectively.

\begin{figure}
\centering \subfigure{\includegraphics[width=5cm]{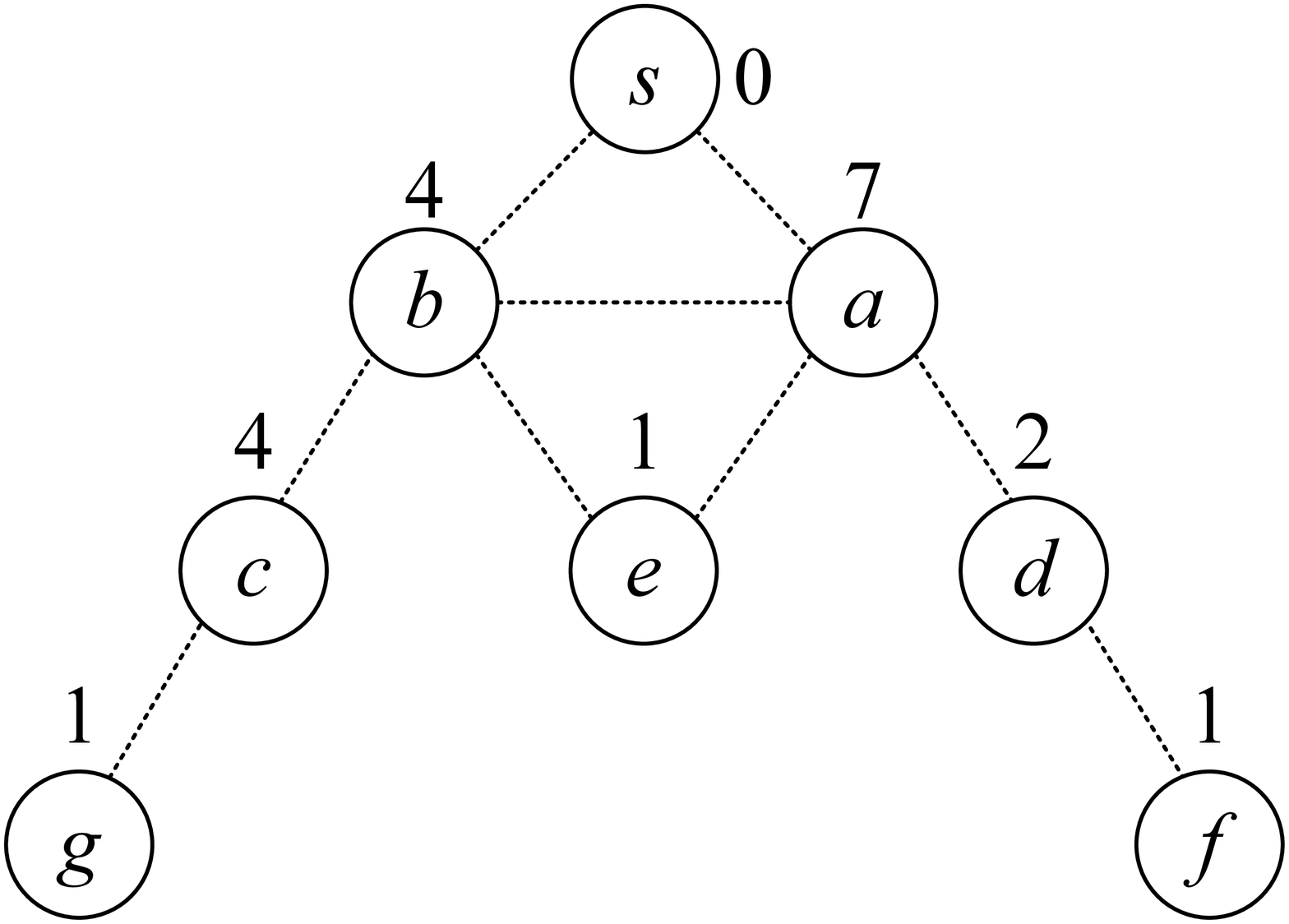}}\\
\subfigure{\includegraphics[width=2.5cm]{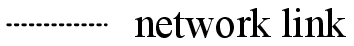}}
\caption{Example of a connected weighted graph $G =
(V_G,E_G,\rho_G)$, where node $s$ is a sink, and the number close to
each node represents the number of units of raw data generated by
the corresponding node.}\label{Fig:net_model_org}
\end{figure}

In the WSN, because the transmission range $R_t$ is limited, sensors
are often hard to communicate with the sink directly, and the data
generated from sensors often have to be forwarded to the sink via
multiple sensors. When raw data are allowed to be aggregated into
packets, it may have fewer packets required to be forwarded, and
thus, reporting data to the sink becomes more efficient. Here,
by the data aggregation model in \cite{8051077}, we assume that at most
$\alpha \in {{\mathbb{Z}}^{+}}$ units of raw data are allowed to be
aggregated into one unit-size packet, where $\alpha$ is also called
aggregation ratio in this paper. Let $\phi(v)$ denote the number of
units of raw data that are required to be forwarded by sensor $v$.
The number of unit-size packets required to be forwarded by $v$,
denoted by $\delta(v)$, is defined as follows:
\begin{equation}\label{eq:pt0}
\delta (v)=\left\lceil \frac{\phi (v)}{\alpha } \right\rceil.
\end{equation}
Take Fig. \ref{Fig:net_model_org}, for example. It is clear that
sensors $b$ and $c$ generate $4$ and $4$ units of raw data,
respectively, in the beginning of a time period, and therefore,
$\phi(b)$ $=$ $4$ and $\phi(c)$ $=$ $4$ initially. When the
aggregation ratio $\alpha$ is assumed to be $3$, we have that
$\delta(b)$ $=$ $\left\lceil \frac{4}{3} \right\rceil$ $=$ $2$ and
$\delta(c)$ $=$ $\left\lceil \frac{4}{3} \right\rceil$ $=$ $2$. If $c$
forwards a unit-size packet that aggregates three units of raw data
to $b$, $b$ will have $4$ $+$ $3$ $=$ $7$ units of raw data required
to be forwarded, and we have that $\phi(b)$ $=$ $7$ and $\delta(b)$
$=$ $\left\lceil \frac{7}{3} \right\rceil$ $=$ $3$. In addition, we
also have that $\phi(c)$ $=$ $4$ $-$ $3$ $=$ $1$ and $\delta(c)$ $=$
$\left\lceil \frac{1}{3} \right\rceil$ $=$ $1$.

In the WSN, by the time division multiple access (TDMA), we assume
that sensors are synchronized, and time is divided into time slots
such that a unit-size packet can be transmitted successfully from
one sensor to its neighboring sensor within one time slot if no data
collision occurs \cite{6195676,Ergen2010,6686433}. In these studies, each sensor is assumed to
have a multi-channel half-duplex transceiver such that each
sensor can switch channels and use one of the channels to transmit
data to another sensor using the same channel at a given time slot.
In addition, each sensor cannot transmit and receive data
simultaneously, and cannot receive data from multiple sensors in the
same time slot. In the WSN, when multiple sensors use the same
channel to transmit data at the same time slot, data collisions may
occur such that the data have to be resent by using one or more time
slots. Here, a data collision will occur at one sensor if the sensor
attempts to transmit two or more packets to multiple sensors, to
receive two or more packets from multiple sensors, to transmit and
receive packets, or to receive one packet and hears another
one at the same time slot
\cite{Ergen:2010:TSA:1824986.1825016,7046421}, which is formally
defined in Definition \ref{def:collision}.

\begin{defi}\label{def:collision}
A data collision is said to be occurred at node $u$ using channel
$ch$ at time slot $t$ if one of the following conditions satisfied: (C1) $u$
transmits two or more packets to sensors at $t$, (C2)
$u$ receives two or more packets from sensors at $t$,
(C3) $u$ transmits and receives packets at $t$, and (C4)
$u$ receives a packet from a sensor $v$ and hears another one from
sensor $w$ at $t$, where $v$ and $w$ use the same channel $ch$.
\end{defi}

Take Fig. \ref{fig_collision}, for example. Figs.
\ref{fig_collision1} and \ref{fig_collision2} show two examples of
data collisions occurred at node $b$. In Fig. \ref{fig_collision1},
when node $b$ receives data from node $a$ and sends data to node $c$
at the same time slot, by the C3 in Definition \ref{def:collision},
there is a data collision occurred at node $b$. In Fig.
\ref{fig_collision2}, when nodes $a$ and $c$ use the same channel to
send data to nodes $b$ and $d$, respectively, at the same time slot,
because $b$ is within the transmission range of $c$, $b$ will hear
the data transmission from $c$, which incurs a data collision by the
C4 in Definition \ref{def:collision}. Note that if $a$ and $c$ use
different channels to send data to nodes $b$ and $d$, respectively,
because the channel of $b$ is the same with that of $a$ for
receiving data from $a$, the channel of $b$ is different from that
of $c$, and thus, $b$ will not hear the data transmission from $c$
and no collision occurs.

\begin{figure}
\center
\subfigure[]{\includegraphics[width=4.7cm]{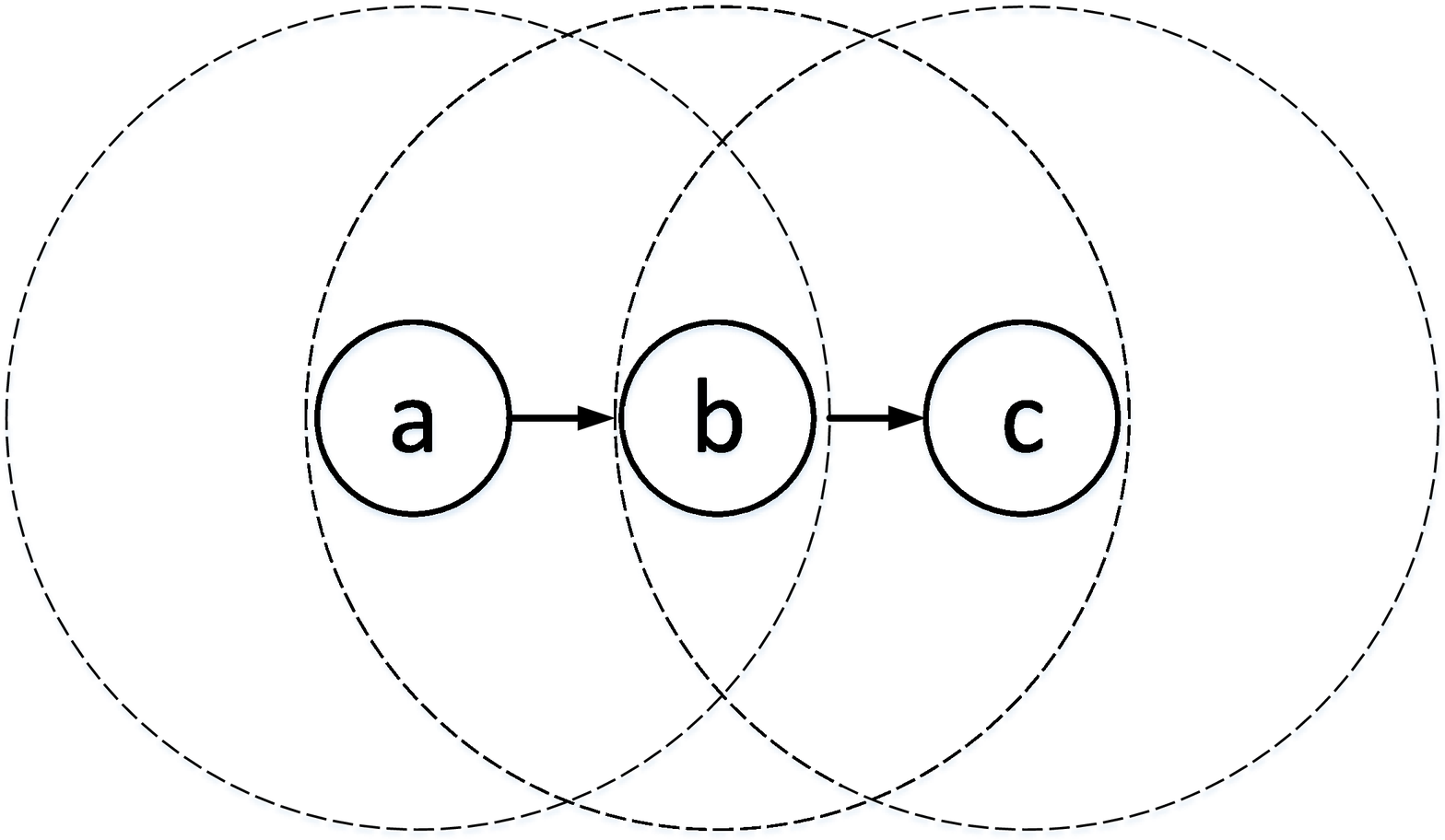}\label{fig_collision1}}
\subfigure[]{\includegraphics[width=5.5cm]{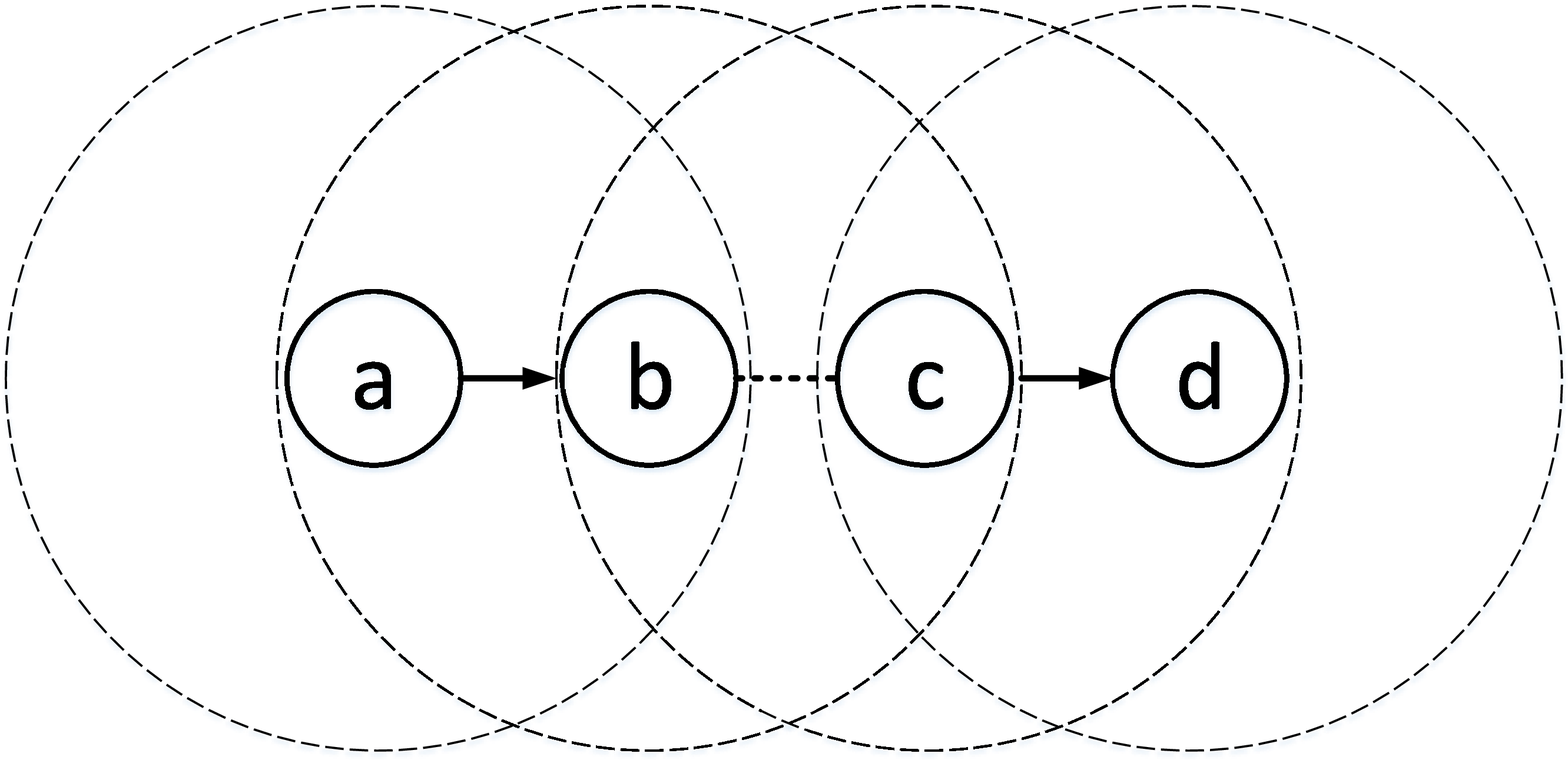}\label{fig_collision2}}
\caption{Examples of data collisions that occur at node $b$. (a)
shows the data collision when $b$ transmits and receives messages at
the same time. (b) shows the data collision when $b$ receives one
message from node $a$ and hears another message from node $c$ at the
same time.} \label{fig_collision}
\end{figure}

\subsection{Minimum-Latency Collision-Avoidance Multiple-Data-Aggregation Scheduling with Multi-Channel Problem and Its Difficulty} \label{scheduling}
In this paper, while given a WSN $G = (V_G,E_G,\rho_G)$, in which
sensors have $\varsigma$ channels to switch, and an aggregation ratio
$\alpha$, our problem is to find a schedule of forwarding data to
the sink without data collisions such that the number of required
time slots is minimized, termed the Minimum-Latency
Collision-Avoidance Multiple-Data-Aggregation Scheduling with
Multi-Channel (MLCAMDAS-MC) problem. Here, the schedule of
forwarding data is a sequence of collision-avoidance schedules
$S_1$, $S_2$, $\ldots$, $S_\ell$, where $S_i$ ($1 \leq i \leq \ell$)
is a set of 3-tuple elements ($u \rightarrow v$, $\gamma$, $ch_j$)
with $u$,$v \in V_G$, $0 < \gamma \leq \alpha$, and $0 < j \leq
\varsigma$, and represents a schedule of unit-size packets at time slot
$i$ without data collisions. In addition, each element ($u
\rightarrow v$, $\gamma$, $ch_j$) $\in S_i$ denotes that one
unit-size packet that aggregates $\gamma$ units of raw data is
scheduled to be forwarded from $u$ to $v$ by using the $j$-th
channel. The MLCAMDAS-MC problem is formally illustrated as follows:

INSTANCE: Given a graph $G = (V_G, E_G, \rho_G)$, an aggregation
ratio $\alpha \in {{\mathbb{Z}}^{+}}$, the total number of channels
$\varsigma$, and a constant $k \in {{\mathbb{Z}}^{+}}$.

QUESTION: Does there exist a schedule of forwarding data, that is, a
sequence of collision-avoidance schedules $S_1$, $S_2$, $\ldots$,
$S_\ell$, for forwarding all generated data to the sink, such that
the number of required time slots $\ell$ is not greater than $k$?


Take Fig. \ref{Fig:net_model_org}, for example. We assume that the
aggregation ratio $\alpha$ = $3$ and the total number of channels
$\varsigma$ $=$ $2$. Also assume that the sequence of
collision-avoidance schedules are $S_1$, $S_2$, $S_3$, $S_4$, $S_5$,
$S_6$, $S_7$, and $S_8$, where ${{S}_{1}}$ $=$ $\{ (c \to b, 3,
ch_1)$, $(f \to d, 1, ch_1)$, $(a \to s, 3, ch_2) \}$, ${{S}_{2}}$
$=$ $\{ (g \to c, 1, ch_1)$, $(b \to s, 3, ch_2)$, $(e \to a, 1,
ch_1) \}$, ${{S}_{3}}$ $=$ $\{ (c \to b, 2, ch_1)$, $(d \to a, 3,
ch_1) \}$, ${{S}_{4}}$ $=$ $\{ (a \to s, 3, ch_1) \}$, ${{S}_{5}}$
$=$ $\{ (b \to s, 3, ch_1) \}$, ${{S}_{6}}$ $=$ $\{ (a \to s, 3,
ch_1) \}$, ${{S}_{7}}$ $=$ $\{ (b \to s, 3, ch_1) \}$, and
${{S}_{8}}$ $=$ $\{ (a \to s, 2, ch_1) \}$. Note that all generated
data can be sent to the sink $s$ within $8$ time slots without any
data collisions.

The difficulty of the MLCAMDAS-MC problem is provided in Theorem
\ref{thm:hardness}.

\begin{thm}\label{thm:hardness}
The MLCAMDAS-MC problem is NP-complete.
\end{thm}

\begin{proof}
It is clear that the MLCAMDAS-MC problem belongs to the NP class. It
suffices to show that the MLCAMDAS-MC problem is NP-hard. Here, the
Minimum-Latency Collision-Avoidance Multiple-Data-Aggregation
Scheduling (MLCAMDAS) problem \cite{8051077} is used to show the difficulty
of the MLCAMDAS-MC problem. The MLCAMDAS problem is formally
illustrated as follows:

INSTANCE: Given a graph $G = (V_G, E_G, \rho_G)$, an aggregation
ratio $\alpha \in {{\mathbb{Z}}^{+}}$, and a constant $k \in
{{\mathbb{Z}}^{+}}$.

QUESTION: Does there exist a sequence of collision-avoidance
schedules $S_1$, $S_2$, $\ldots$, $S_\ell$ for forwarding all
generated data to the sink such that the number of required time
slots $\ell$ is not greater than $k$?

It is clear that when the total number of channels $\varsigma$ $=$ $1$,
the MLCAMDAS-MC problem is equivalent to the MLCAMDAS problem, which
implies that the MLCAMDAS problem is a subproblem of the MLCAMDAS-MC
problem. Because the MLCAMDAS problem is NP-hard \cite{8051077}, and
therefore, the MLCAMDAS-MC problem is also NP-hard. This thus
completes the proof.
\end{proof}

\section{Extended Relative Collision Graph}\label{relative}
Because the MLCAMDAS-MC problem is to find a schedule of forwarding data to
the sink without data collisions, determining collision relation between any data transmission
is important to the MLCAMDAS-MC problem. To discover the collision relation, the concept of the relative collision graph
is borrowed from the research in \cite{8051077} and is extended here
to represent the collision relation between any data transmission for the MLCAMDAS-MC problem.
When each sensor in the WSN has exactly one channel, to minimize the
number of required time slots for collecting the generated data,
all possible data transmissions that forward data to nodes closer to the sink are considered to construct
the relative collision graph \cite{8051077}. For this purpose,
a directed graph $G_\zeta = (V_{G_\zeta}, E_{G_\zeta})$, termed the data-forwarding graph,
is used to represent all such possible data transmissions in the WSN
$G = (V_G, E_G, \rho_G)$, where $V_{G_\zeta}$ $=$ $V_G$, an edge $(u,v)$ is included in
$E_{G_\zeta}$ if $(u,v) \in E_G$ and $hop(u) > hop(v)$, and $hop(u)$ (or $hop(v)$) denotes the minimum hop count from node $u$ (or $v$) to the sink in
$G$. Take Fig.
\ref{Fig:previous:forwarding_graph}, for example. Fig. \ref{Fig:previous:forwarding_graph} illustrates
the data-forwarding graph $G_\zeta$ of the WSN $G$ shown in Fig. \ref{Fig:net_model_org}.
Because $hop(e)$ $=$ $2$ $>$ $hop(a)$ $=$ $hop(b)$ $=$ $1$ and edges $(e,b)$, $(e,a)$ $\in E_G$ in Fig. \ref{Fig:net_model_org}, edges $(e,a)$ and $(e,b)$ are included
in $E_{G_\zeta}$. In addition, although edge $(a,b)$ $\in$ $E_G$ in Fig. \ref{Fig:net_model_org}, edge $(a,b)$ or $(b,a)$ is not included in
$E_{G_\zeta}$ because $hop(a)$ (or $hop(b)$) is not greater than $hop(b)$ (or $hop(a)$).

\begin{figure}
\centering
\subfigure[]{\includegraphics[width=4cm]{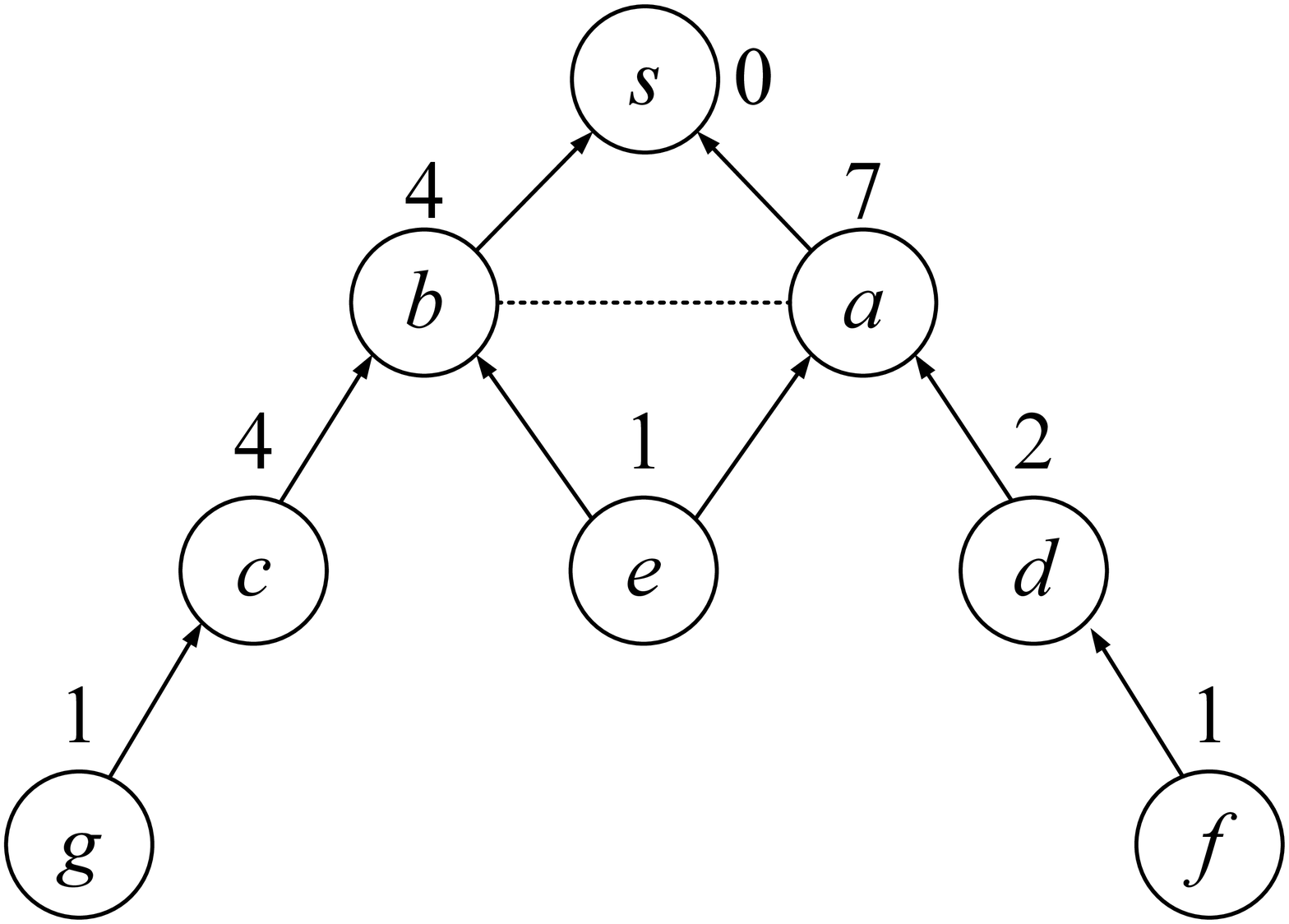}\label{Fig:previous:forwarding_graph}}
\subfigure[]{\includegraphics[width=3.5cm]{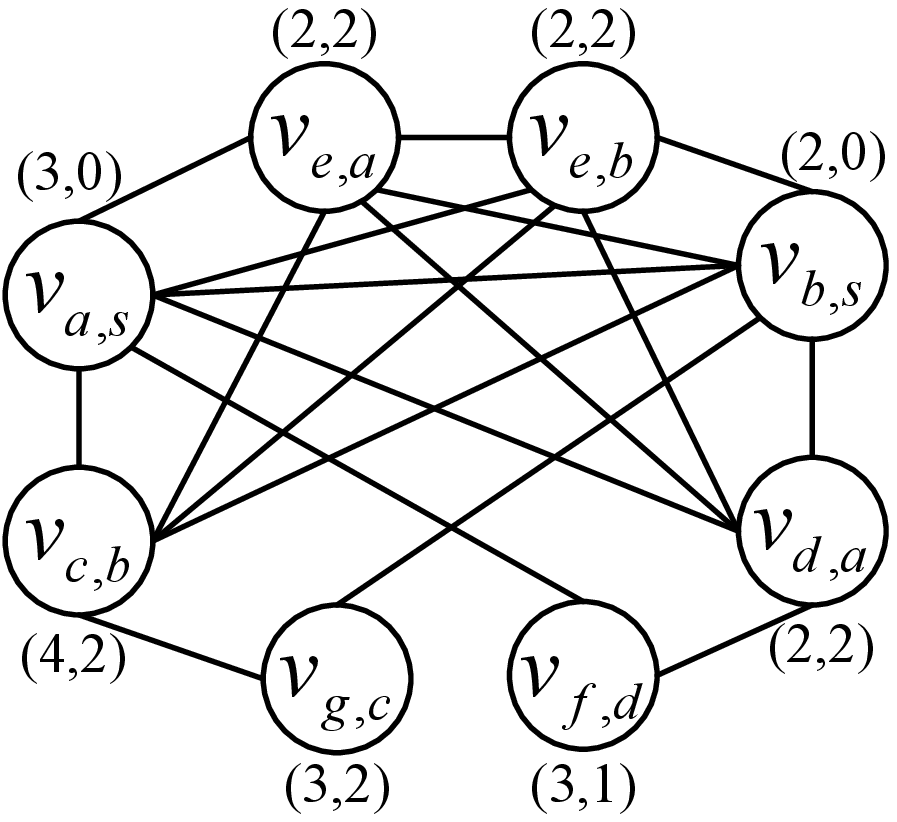}\label{Fig:previous:relative}}\\
\subfigure{\includegraphics[width=8cm]{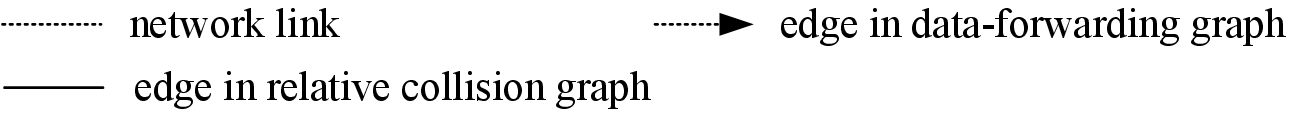}}
\caption{Examples of the data-forwarding graph $G_\zeta$
and the relative collision graph $G_r$ of the WSN $G$ shown in Fig. \ref{Fig:net_model_org}, where the aggregation ratio $\alpha$ is assumed to be
$3$.
(a) $G$'s corresponding data-forwarding graph $G_\zeta = (V_{G_\zeta}, E_{G_\zeta})$. (b) $G$'s corresponding
relative collision graph $G_r = (V_{G_r},
E_{G_r}, \omega_{G_r}, \eta_{G_r})$, where the left number and the
right number in parentheses close to each node $v \in V_{G_r}$
represents $\omega_{G_r}(v)$ and $\eta_{G_r}(v)$, respectively.
}\label{Fig:previouswork}
\end{figure}

By the WSN $G$ and its corresponding data-forwarding graph $G_\zeta$,
the corresponding relative collision graph $G_r = (V_{G_r}, E_{G_r},
\omega_{G_r}, \eta_{G_r})$ can be constructed
to illustrate the collision relation between any data transmission in \cite{8051077}.
In $G_r$, each node $v_{x,y} \in V_{G_r}$ represents a possible data transmission from node $x$ to node $y$
in $G_\zeta$, that is, an edge $(x,y) \in G_\zeta$. Each edge $(v_{x,y}, v_{z,w}) \in
E_{G_r}$ represents that a data collision occurs if data are sent from node $x$ to node $y$ and from node
$z$ to node $w$ at the same time slot and in the same channel.
In addition, for each $v_{x,y} \in V_{G_r}$, $\omega_{G_r}(v_{x,y})$ calculated by $\delta(x) \times hop(x)$
is used to represent the weight of node $v_{x,y}$; and $\eta_{G_r}(v_{x,y})$ calculated by $\alpha$ $\times$ $\left\lceil
\frac{{\phi(y)}}{\alpha } \right\rceil$ $-$ $\phi(y)$ is used to represent
how many extra units of raw data can be aggregated
with the data at $y$ such that the total number of aggregated
unit-size packets in $y$ is not increased. When $G$ and $G_\zeta$ are given,
the corresponding $G_r = (V_{G_r}, E_{G_r}, \omega_{G_r}, \eta_{G_r})$
is constructed as follows: $V_{G_r}$ is the set of nodes $v_{x,y}$ for
each $(x,y) \in E_{G_\zeta}$; $E_{G_r}$ is the set of edges
$(v_{x,y},v_{z,w})$ if edges $(x,y), (z,w) \in E_{G_\zeta}$ and a data collision occurs under the condition that data are transmitted from node $x$
to node $y$ and from node $z$ to node $w$ at the same time slot and in the same channel; and
$\omega_{G_r}(v_{x,y})$ $=$ $\delta(x) \times hop(x)$ and $\eta_{G_r}(v_{x,y})$ $=$ $\alpha$ $\times$ $\left\lceil
\frac{{\phi(y)}}{\alpha } \right\rceil$ $-$ $\phi(y)$ for each
$v_{x,y} \in V_{G_r}$. Take Fig. \ref{Fig:previous:relative}, for example.
When given the WSN $G$, as shown in Fig. \ref{Fig:net_model_org}, and the corresponding
data-forwarding graph $G_\zeta$, as shown in Fig.
\ref{Fig:previous:forwarding_graph}, when the aggregation ratio $\alpha$ $=$
$3$, the corresponding relative collision graph $G_r = (V_{G_r}, E_{G_r},
\omega_{G_r}, \eta_{G_r})$ is shown in Fig. \ref{Fig:previous:relative}.
Note that nodes $v_{e,a}$
and $v_{e,b}$ are included in $V_{G_r}$ because edges $(e,a)$, $(e,b)$ $\in$ $E_{G_\zeta}$.
In addition, edge $(v_{e,a}, v_{e,b})$ $\in$ $E_{G_r}$ because a data collision occurs
by the C1 in Definition \ref{def:collision} when data are transmitted from node $e$
to node $a$ and from node $e$ to node $b$ in the same time slot. Because $e.\delta$ $=$ $\left\lceil \frac{{\phi(e)}}{\alpha}
\right\rceil$ $=$ $\left\lceil \frac{{1}}{3} \right\rceil$ $=$ $1$
and $e.hop = 2$, we have that $\omega_{G_r}(v_{e,a})$ $=$
$\omega_{G_r}(v_{e,b})$ $=$ $1 \times 2$ $=$ $2$,
$\eta_{G_r}(v_{e,a})$ $=$ $\alpha$ $\left\lceil
\frac{{\phi(a)}}{\alpha } \right\rceil$ $-$ $\phi(a)$ $=$ $3 \left\lceil
\frac{{7}}{3} \right\rceil$ $-$ $7$ $=$ $2$, and
$\eta_{G_r}(v_{e,b})$ $=$ $\alpha$ $\left\lceil
\frac{{\phi(b)}}{\alpha } \right\rceil$ $-$ $\phi(b)$ $=$ $3 \left\lceil
\frac{{4}}{3} \right\rceil$ $-$ $4$ $=$ $2$.

Because the relative collision graph only considers exactly one channel used by every sensor in the WSN,
the extended relative collision graph is therefore presented to represent collision relation with multiple channels
for the MLCAMDAS-MC problem. When sensors have $\varsigma$ channels to switch, each sensor in the WSN
can use the $i$-th ($1$ $\leq$ $i$ $\leq$ $\varsigma$) channel for data transmission.
When all sensors in the WSN use the same $i$-th ($1$ $\leq$ $i$ $\leq$ $\varsigma$) channel,
the collision relation for all possible data transmission is the same as that in the relative collision graph.
In addition, because sensors in the WSN can select one of $\varsigma$ channels for data transmission,
the collision relation between any data transmission using different channels has to be considered in the extended
relative collision graph. By the observations, the extended relative collision graph is constructed
by including $\varsigma$ relative collision graphs each representing the collision relation for all possible data transmission
with channel $i$ ($1$ $\leq$ $i$ $\leq$ $\varsigma$). In addition, some edges are inserted between nodes in different
relative collision graphs to represent the collision relation for the data transmission using different channels.
When $G$ and $G_\zeta$ are given,
the corresponding extended relative collision graph $G_{r+} = (V_{G_{r+}}, E_{G_{r+}}, \omega_{G_{r+}}, \eta_{G_{r+}})$
is constructed as follows:  $V_{G_{r+}}$ is the set of nodes $v_{x,y}^i$ for
each $(x,y) \in E_{G_\zeta}$ and $1$ $\leq$ $i$ $\leq$ $\varsigma$; $E_{G_{r+}}$ is the set of edges
$(v_{x,y}^i,v_{z,w}^j)$ if edges $(x,y), (z,w) \in E_{G_\zeta}$ and a data collision occurs under the condition that data are transmitted from node $x$ to node $y$ by using channel $i$ and from node $z$ to node $w$ by using channel $j$ at the same time slot;
and $\omega_{G_{r+}}(v_{x,y}^i)$ $=$ $\delta(x) \times hop(x)$ and $\eta_{G_{r+}}(v_{x,y}^i)$ $=$ $\alpha$ $\times$ $\left\lceil
\frac{{\phi(y)}}{\alpha } \right\rceil$ $-$ $\phi(y)$ for each
$v_{x,y}^i \in V_{G_{r+}}$. Note that $\omega_{G_{r+}}$ and  $\eta_{G_{r+}}$ in $G_{r+}$ have the same definition as that in
$G_{r}$.

\begin{figure}
\centering
\subfigure{\includegraphics[width=4.5cm]{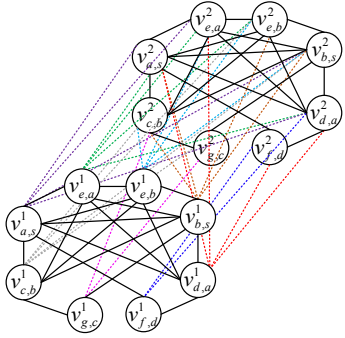}}
\caption{Example of the extended relative collision graph $G_{r+}$ $=$ $(V_{G_{r+}},
E_{G_{r+}}, \omega_{G_{r+}}, \eta_{G_{r+}})$ of the WSN $G$ shown in Fig. \ref{Fig:net_model_org}, where the aggregation ratio $\alpha$ and the number of channels $\varsigma$ are assumed to be $3$ and $2$, respectively.
}
\label{fig:example:ERCG}
\end{figure}

Take Fig. \ref{fig:example:ERCG}, for example.
When the WSN $G$ shown in Fig. \ref{Fig:net_model_org},
the corresponding data-forwarding graph $G_\zeta$ shown in Fig.
\ref{Fig:previous:forwarding_graph}, the aggregation ratio $\alpha$ $=$
$3$, and the number of channels $\varsigma$ $=$ $2$ are given,
the corresponding extended relative collision graph $G_{r+} = (V_{G_{r+}}, E_{G_{r+}},
\omega_{G_{r+}}, \eta_{G_{r+}})$ is shown in Fig. \ref{fig:example:ERCG}.
Because the number of channels $\varsigma$ $=$ $2$, it is clear that the $G_{r+}$ in Fig. \ref{fig:example:ERCG}
consists of two relative collision graphs that each are the same with the $G_r$ as shown in Fig. \ref{Fig:previous:relative}.
In addition, edge $(v_{e,a}^{1}, v_{e,b}^{2})$ $\in$ $E_{G_{r+}}$ because a data collision occurs at node $e$
by the C1 in Definition \ref{def:collision} when $e$ transmits packets to node $a$ and $b$ at the same time slot.
Edge $(v_{e,a}^{1}, v_{d,a}^{2})$ $\in$ $E_{G_{r+}}$ because a data collision occurs at node $a$
by the C2 in Definition \ref{def:collision} when $a$ receives packets from $e$ and $d$ at the same time slot.
Edge $(v_{e,a}^{1}, v_{e,a}^{2})$ $\in$ $E_{G_{r+}}$ because data collisions occur at nodes $e$ and $a$
by the C1 and C2 in Definition \ref{def:collision}.
Edge $(v_{e,a}^{1}, v_{a,s}^{2})$ $\in$ $E_{G_{r+}}$ because a data collision occurs at node $a$
by the C3 in Definition \ref{def:collision} when $a$ receives a packet from $e$ and transmit another one to $s$
at the same time slot.
Moreover, note that although $(v_{e,a}^{1}, v_{b,s}^{1})$ $\in$ $E_{G_{r+}}$ by the C4 in Definition \ref{def:collision},
that is, a data collision occurs at node $a$ when $a$ receives a packet from node $e$ and hears another one from
node $b$ by using the first channel at the same time slot,
$(v_{e,a}^{1}, v_{b,s}^{2})$ $\notin$ $E_{G_{r+}}$ because $a$ and $b$ use different channels to receive and transmit data, respectively.

\section{Distributed Collision-Avoidance Scheduling (DCAS) Algorithm} \label{sec:proposed_algorithm}
By the extended relative collision graph $G_{r+} = (V_{G_{r+}}, E_{G_{r+}}, \omega_{G_{r+}}, \eta_{G_{r+}})$,
edge $(v_{x,y}^i, v_{z,w}^j) \notin E_{G_{r+}}$ ($1$ $\leq$ $i$, $j$ $\leq$ $\varsigma$) represents that a data collision will not occur if data are transmitted from node $x$ to node $y$ with the $i$-th channel and from node
$z$ to node $w$ with the $j$-th channel at the same time slot. This implies that
if $IS$ is an independent set in $G_{r+}$, that is, $IS \subseteq V_{G_{r+}}$ and edge $(u,v) \notin E_{G_{r+}}$ for any $u,v \in IS$,
there is no data collision when data are transmitted from node $x$ to node $y$ with the $i$-th channel for all $v_{x,y}^i \in IS$ at the same time slot. Therefore, to avoid data collisions,
the idea is to select a suitable independent set from the extended relative collision graph $G_{r+}$ for each time slot. Take Fig. \ref{fig:example:ERCG}, for example. Let $IS$ $=$ $\{v_{c,b}^1$, $v_{f,d}^1$, $v_{a,s}^2 \}$. Clearly,
$IS$ is an independent set in $G_{r+}$ shown in Fig. \ref{fig:example:ERCG}.
It is also clear that no data collision will occur in Fig. \ref{Fig:net_model_org} when data are transmitted from node $c$ to node $b$ with the first channel,
from node $f$ to node $d$ with the first channel, and from node $a$ to node $s$ with the second channel at the same time slot.

To find a suitable independent set from the extended relative collision graph $G_{r+}$ for each time slot,
the idea is to find an independent set $IS$ in $G_{r+}$ that is composed of the nodes $v_{x,y}^i$ with higher
$\omega_{G_r}(v_{x,y}^i)$ such that the nodes $v_{x,y}^i$ with higher number of unit-size packets required to be forwarded by $x$, that is, $\delta(x)$, or higher minimum hop count from $x$ to the sink, that is, $hop(x)$,
can be selected to minimize the total required time slots,
where $\omega_{G_{r+}}(v_{x,y}^i)$ $=$ $\delta(x) \times hop(x)$.
To select suitable nodes to form an independent set in $G_{r+}$,
the precedence of nodes, which is used to decide which node has precedence to be
selected into the independent set, has to be determined first.
Here, for any two nodes $v_{x,y}^i$ and $v_{z,w}^j$ in $G_{r+}$,
$v_{x,y}^i$ is said to have higher precedence than $v_{z,w}^j$ if $\omega_{G_{r+}}(v_{x,y}^i)$ $>$
$\omega_{G_{r+}}(v_{z,w}^j)$; otherwise, if $\omega_{G_{r+}}(v_{x,y}^i)$ is equal to
$\omega_{G_{r+}}(v_{z,w}^j)$, the node with higher $\eta_{G_{r+}}$ value
has higher precedence because at most $\eta_{G_{r+}}(v_{x,y}^i)$ (or $\eta_{G_{r+}}(v_{z,w}^j)$) units of raw data can be aggregated at $y$ (or $w$) without increasing $\delta(y)$ (or $\delta(w)$); otherwise, if $\eta_{G_{r+}}(v_{x,y}^i)$
is equal to $\eta_{G_{r+}}(v_{z,w}^j)$, the node with smaller channel number has higher precedence;
otherwise, if channel $i$ is equal to channel $j$, the node with higher ID value ($ID(v_{x,y}^i)$ or $ID(v_{z,w}^j)$) has higher precedence,
where $ID(v_{a,b}^c)$ is a pair of $id(a)$ and $id(b)$, denoted by ($id(a)$, $id(b)$), for each $v_{a,b}^c \in G_{r+}$; $id(v)$
is assumed to be an unique identification for each $v \in V_G$; and ($id(a)$, $id(b)$) is said to be higher than ($id(c)$, $id(d)$)
if $id(a)$ $>$ $id(c)$, or ($id(a)$ $=$ $id(c)$ and $id(b)$ $>$ $id(d)$.
The definition of the precedence of nodes is formally defined in Definition \ref{def:precedence}.

\begin{defi}\label{def:precedence}
Given an extended relative collision graph $G_{r+} = (V_{G_{r+}}, E_{G_{r+}},
\omega_{G_{r+}}, \eta_{G_{r+}})$, node $v_{x,y}^i$ $\in$ $V_{G_r}$ is said to have higher precedence
over $v_{z,w}^j$ $\in$ $V_{G_r}$ ($v_{z,w}^j$ $\neq$ $v_{x,y}^i$) if
$(\omega_{G_{r+}}(v_{x,y}^i) > \omega_{G_{r+}}(v_{z,w}^j))$ or
$(\omega_{G_{r+}}(v_{x,y}^i) = \omega_{G_{r+}}(v_{z,w}^j)$ and $\eta_{G_{r+}}(v_{x,y}^i) > \eta_{G_{r+}}(v_{z,w}^j))$
or $(\omega_{G_{r+}}(v_{x,y}^i) = \omega_{G_{r+}}(v_{z,w}^j)$ and $\eta_{G_{r+}}(v_{x,y}^i) = \eta_{G_{r+}}(v_{z,w}^j)$ and $(i < j))$ or
$(\omega_{G_{r+}}(v_{x,y}^i) = \omega_{G_{r+}}(v_{z,w}^j)$ and $\eta_{G_{r+}}(v_{x,y}^i) = \eta_{G_{r+}}(v_{z,w}^j)$ and $(i = j)$ and $ID(v_{x,y}^i) > ID(v_{z,w}^j))$, where $ID(v_{x,y}^i)$ is a pair of $id(x)$ and $id(y)$; and $id(v)$ denotes $v$'s identification for each $v$ $\in$ $G$.
\end{defi}

To design a distributed algorithm for the MLCAMDAS-MC problem, every sensor $u$
is assumed to have the information about the minimum hop count from $u$ to the sink, that is, $hop(u)$, used
to evaluate the value of $\omega_{G_{r+}}$. It can be easily achieved by flooding a message from the sink
to all nodes in the networks based on the breadth-first-search mechanism \cite{Nguyen201699}.
In addition, every sensor in the networks is also assumed to have limited
local information. By the observation in Fig.
\ref{fig_collision2}, if node $a$ has local information about nodes $b$, $c$, and $d$, and knows that the data transmission from $c$ to $d$ with channel $j$
is a better choice than the data transmission from $a$ to $b$ with channel $i$, that is, $a$ has a local subgraph of $G_{r+}$ and
knows that $v_{c,d}^j$ has higher precedence than $v_{a,b}^i$, $a$ will let $c$ to schedule the data transmission from $c$ to $d$ with channel $j$
before $a$'s schedule.  In Fig.
\ref{fig_collision2}, we have that $a$ requires at least $3$-hop neighboring information to compare all possible data transmission that could have data collision
with the data transmission from $a$. Therefore, in this paper, we assume that every sensor in the networks maintains $3$-hop neighboring information, that is, every sensor $u$ in the networks has a local WSN $u.G$ $=$
$(u.V_G, u.E_G, \rho_G)$ and a local data-forwarding graph $u.G_\zeta$ $=$ $(u.V_{G_\zeta}, u.E_{G_\zeta})$. By $u.G$ and $u.G_\zeta$, a local extended relative collision graph $u.G_{r+}$ $=$ $(u.V_{G_{r+}}, u.E_{G_{r+}}, \omega_{G_{r+}}, \eta_{G_{r+}})$ can then be constructed. Take Fig. \ref{Fig:subgraphs}, for example.
Fig. \ref{Fig:subgraphs} shows the local information of node $f$ that is a node in $G$ shown in Fig. \ref{Fig:net_model_org}.
Fig. \ref{Fig:sub_forwarding_graph} shows $f$'s local WSN $f.G$ and local data-forwarding graph $f.G_\zeta$.
Note that node $b$ is in $f.G$ because $b$ is within $3$-hop distance from $f$.
Also note that edge $(b,s)$ is not in both of $f.G$ and $f.G_\zeta$ because $(b,s)$ is not within $3$-hop distance from $f$.
Fig. \ref{Fig:sub_ERC} shows $f$'s local extended relative collision graph $f.G_{r+}$.
Note that there are $8$ nodes in $f.G_{r+}$ $=$ $(f.V_{G_{r+}}, f.E_{G_{r+}}, \omega_{G_{r+}}, \eta_{G_{r+}})$ because the number of channels $\varsigma$ is $2$ and four edges exist in $f.G_\zeta$.
Also note that $f.G_{r+}$ is a subgraph of $G_{r+}$ induced by the nodes in $f.V_{G_{r+}}$, where $G_{r+}$ is shown in Fig. \ref{fig:example:ERCG}.

\begin{figure}
\centering
\subfigure[]{\includegraphics[width=3cm]{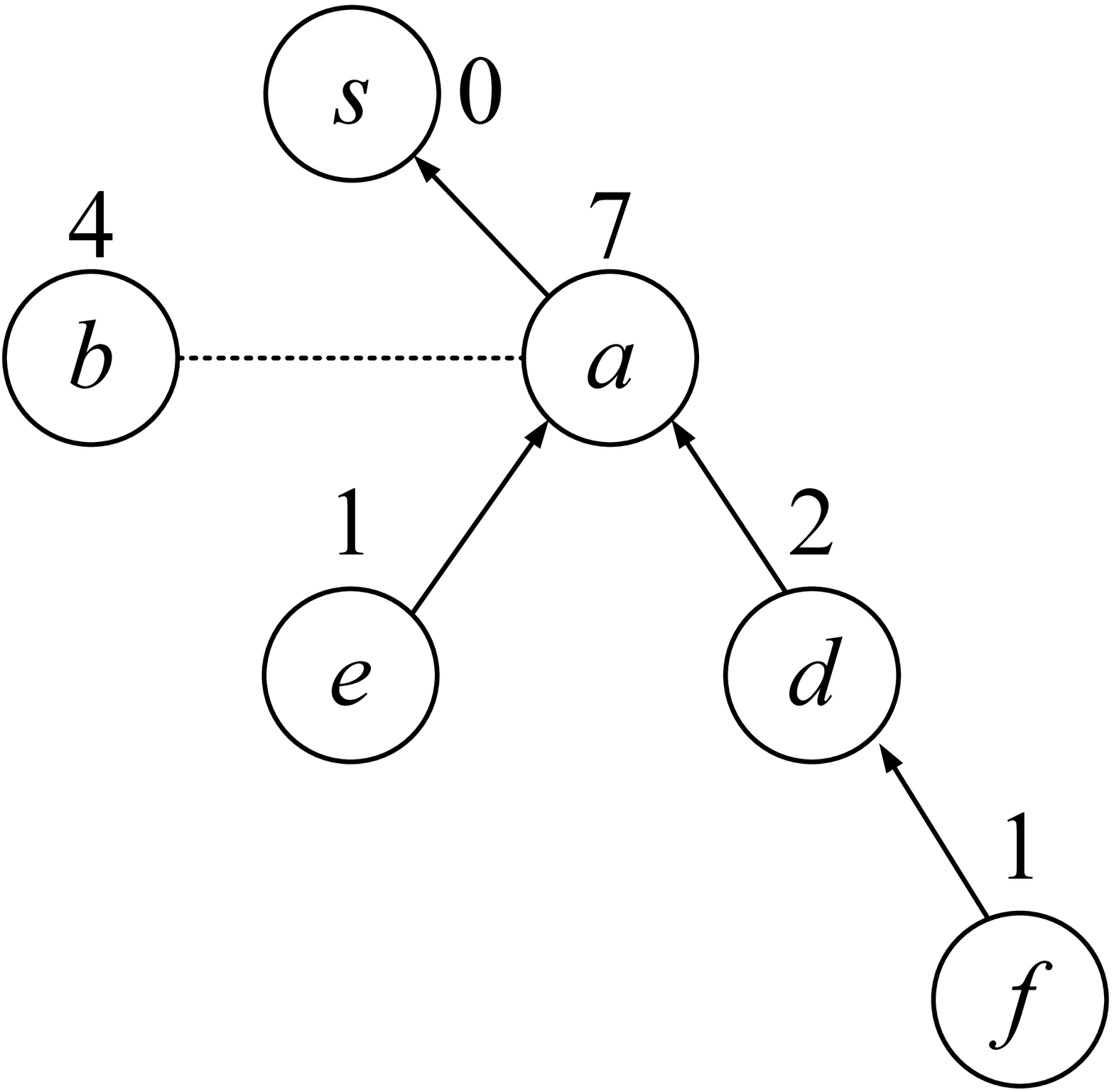}\label{Fig:sub_forwarding_graph}}
\subfigure[]{\includegraphics[width=4.5cm]{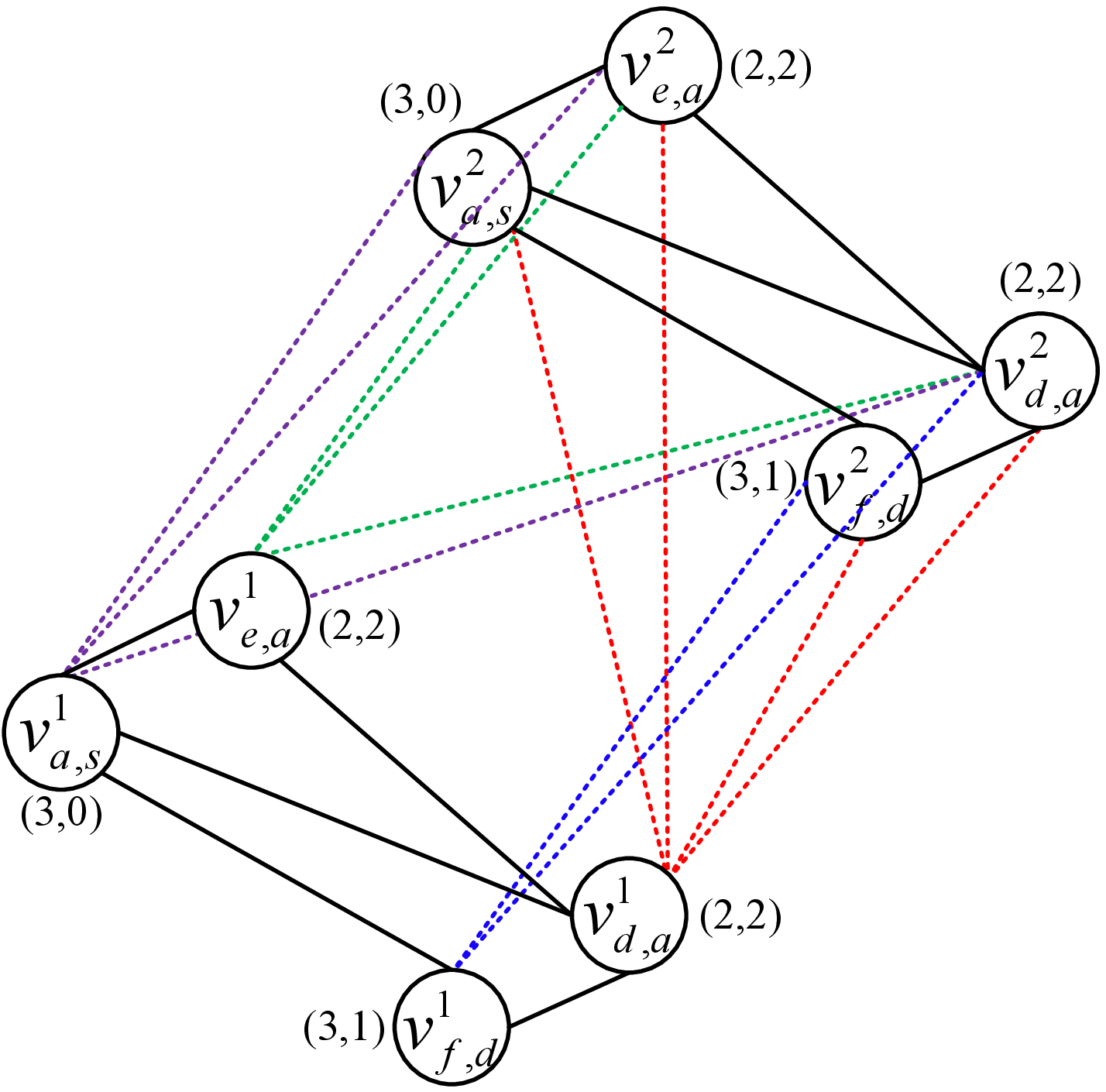}\label{Fig:sub_ERC}}\\
\caption{Examples of node $f$'s local information, including local WSN $f.G$, local data-forwarding graph $f.G_\zeta$,
and local extended relative collision graph $f.G_{r+}$, where $f$ is a node in $G$ shown in Fig. \ref{Fig:net_model_org}.
(a) The combination of $f.G$ and $f.G_\zeta$. (b) $f.G_{r+}$ $=$ $(f.V_{G_{r+}}, f.E_{G_{r+}}, \omega_{G_{r+}}, \eta_{G_{r+}})$, where the aggregation ratio $\alpha$ and the number of channels $\varsigma$ are assumed to be $3$ and $2$, respectively; and the left number and the right number in parentheses close to each node $v \in f.V_{G_{r+}}$
represents $\omega_{G_{r+}}(v)$ and $\eta_{G_{r+}}(v)$, respectively.
}\label{Fig:subgraphs}
\end{figure}

When every sensor maintains 3-hop neighboring information, the idea of the proposed distributed algorithm, termed the
distributed collision-avoidance scheduling (DCAS) algorithm, is to let every sensor $u$ iteratively schedule local data transmissions
for each time slot by its 3-hop neighboring information, until no more data are required to be scheduled by $u$.
Here, we use $u.t$ to denote which time slot waited to be scheduled by sensor $u$, and use $u.S$ to store the determined schedules.
Initially, for each sensor $u$ in the networks, $u.t$ is set to $1$, and $u.S$ is set to $\emptyset$. When some sensor $u$ makes a schedule for time slot $u.t$,
$u.t$ is incremented by $1$, which is used to represent that $u$ is ready for the next time slot.
In addition, when some sensor $v \in u.V_G$ satisfies that $v.t < u.t$, it implies that $v$ has not yet made a decision for time slot $v.t$,
and thus, $u$ cannot make any schedule due to the lack of $v$'s information at time slot $u.t$.
Therefore, $u$ can make a schedule if $u.t \le v.t$ for all $v \in u.V_G$.

When $u.t \le v.t$ for all $v \in u.V_G$, sensor $u$ is checked to see if it is allowed to make a schedule by
Procedure SCHEDULE. In Procedure SCHEDULE, the idea is that sensor $u$ can make a schedule of transmitting
data from $u$ to its neighboring sensor $y$ with channel $i$ if the data transmission from $u$ to $y$ with channel $i$,
represented by $v_{u,y}^{i}$, has higher precedence than other possible data transmissions $v_{z,w}^{j}$ in $u.V_{G_{r+}}$ by Definition \ref{def:precedence}. When $u$ can make a schedule of transmitting
data from $u$ to $y$ with channel $i$, the schedule is inserted into $u.S$, and a MSG$\_$DECISION message with
the scheduling information is locally broadcast to all sensors in $u.V_G$. After this, $u.t$ is incremented by $1$.
When other sensor receives the MSG$\_$DECISION message, it will locally update the related information.
In addition, if $y$ receives the MSG$\_$DECISION message, the scheduling information is inserted into $y.S$, and the MSG$\_$DECISION message is locally broadcast to all sensors in $y.V_G$.

To avoid data collision, the data transmissions scheduled for the same time slot have to be collision-free, that is,
the nodes $v_{x,y}^{i}$ selected for the same time slot have to form an independent set in $G_{r+}$.
To this purpose, for each sensor $u$ in the networks, all nodes in $u.V_{G_{r+}}$ are marked as white when $u$ changes $u.t$ to a new time slot.
In addition, when any node $v_{u,y}^{i}$ in $V_{G_{r+}}$ is selected by node $u$ for scheduling time slot $t$, $v_{u,y}^{i}$ and its neighboring nodes in $G_{r+}$ are marked as black in the local information of other sensors $u'$ with $u'.t = t$. When $v_{u,y}^{i}$ and its neighboring nodes in $G_{r+}$ are marked as black for time slot $t$, the node selected from the remaining white nodes in $G_{r+}$ will be independent from $v_{u,y}^{i}$ in $G_{r+}$ for time slot $t$. When one sensor $u$ knows that all nodes in $u.V_{G_{r+}}$ are black, it implies that no data transmission can be scheduled
by $u$ for the current time slot. Then, $u$ locally broadcasts a MSG$\_$SKIP message to all sensors in $u.V_G$ about that $u$ skips the current time slot and $u.t$ is incremented by $1$.

After scheduling a number of time slots,
when a sensor $u$ has $\phi(u)$ $=$ $0$ and satisfies $hop(u)$ $\geq$ $hop(v)$ for each of its neighboring sensors $v$, we have that no data will be scheduled to and forwarded by $u$, and thus, $u$ finishes its scheduling process.
In addition, when a sensor $u$ has $\phi(u)$ $=$ $0$, and each of its neighboring sensors $v$ having $hop(v) > hop (u)$ finishes the scheduling process, $u$ will also finish its scheduling process. When a sensor $u$ finishes its scheduling process, a MSG$\_$FINISHED message with $u$
is locally broadcast to all sensors in $u.V_G$.
When the sink knows that all its neighboring sensors finish the scheduling processes,
the data scheduling for the network is completed. The details of the DCAS algorithm is shown in Algorithm \ref{alg:DCAS}.

\begin{algorithm*}
\begin{algorithmic}[1]\label{alg:Schedule}
\Procedure {SCHEDULE} {$u$}

    \State Let $V$ be the set of white nodes $v_{z,w}^{j}$ in $u.V_{G_{r+}}$

    \State Let $v_{x,y}^{i}$ be the node having highest precedence in $V$

    \If {$u = x$}

        \State $\varphi$ $\gets$ $min(\phi(u), \alpha)$, where $min(a,b)$ denotes a minimum function and produces the minimum value of $a$ and $b$

        \State $u.S$ $\gets$ $u.S \bigcup \left\{ S_{u.t} \right\}$, where $S_{u.t}$ $=$ ($x
\rightarrow y$, $\varphi$, $ch_i$)

        \State $\phi(u)$ $\gets$ $\phi(u)$ $-$ $\varphi$; $\phi(y)$ $\gets$ $\phi(y)$ $+$ $\varphi$

        \State update the values of $\omega_{G_{r+}}$ and $\eta_{G_{r+}}$ for the nodes in $u.V_{G_{r+}}$

        \State locally broadcast a MSG$\_$DECISION message with $S_{u.t}$ to all sensors in $u.V_G$

        \State all nodes in $u.V_{G_{r+}}$ become white

        \State $u.t$ $\gets$ $u.t + 1$; $y.t$ $\gets$ $u.t$

    \ElsIf{$V = \emptyset$}

        \State locally broadcast a MSG$\_$SKIP message with $u$ and $u.t$ to all sensors in $u.V_G$

        \State all nodes in $u.V_{G_{r+}}$ become white

        \State $u.t$ $\gets$ $u.t + 1$

    \EndIf

\EndProcedure
\end{algorithmic}
\end{algorithm*}

\begin{algorithm*} \caption{DCAS($u$)}\label{alg:DCAS}
\begin{algorithmic}[1]

    \State $u.S \gets \emptyset$; $u.finished \gets false$; $u.t \gets 1$

    \State Construct $u.G$, $u.G_\zeta$, and $u.G_{r+}$

    \State all nodes in $u.V_{G_{r+}}$ are initialized to be white

    \State $\phi(v)$ is set to $\rho_G(v)$ for all $v \in u.V_G$

    \While {$u.finished = false$}

         \If {a MSG$\_$DECISION message with $S_{t'}$ $=$ ($x \rightarrow y$, $\varphi$, $ch_i$) is received by $u$ for the first time and $u \neq x$}

            \State $\phi(x)$ $\gets$ $\phi(x)$ $-$ $\varphi$; $\phi(y)$ $\gets$ $\phi(y)$ $+$ $\varphi$

            \State update the values of $\omega_{G_{r+}}$ and $\eta_{G_{r+}}$ for the nodes in $u.V_{G_{r+}}$

            \If {$u$ $=$ $y$}

                \State all nodes in $u.V_{G_{r+}}$ become white

                \State $u.S$ $\gets$ $u.S \bigcup \left\{ S_{t'} \right\}$

                \State locally broadcast a MSG$\_$DECISION message with $S_{t'}$ to all sensors in $u.V_G$

            \ElsIf {$u.t$ $=$ $t'$}

                \State $v_{x,y}^{i}$ and its neighboring nodes in $u.V_{G_{r+}}$ become black

            \EndIf

            \State $x.t$ $\gets$ $t'$ $+$ $1$; $y.t$ $\gets$ $x.t$

         \EndIf

         \If {a MSG$\_$SKIP message with $v$ and $t'$ is received by $u$ for the first time and $u \neq v$ }

            \If {$u.t$ $=$ $t'$}

                \State nodes $v_{z,w}^{j}$ become black for all $v_{z,w}^{j}$ $\in$ $u.V_{G_{r+}}$ with $z = v$ or $w = v$

            \EndIf

            \State $v.t$ $\gets$ $t'$ $+$ $1$

         \EndIf

         \If {a MSG$\_$FINISHED message with $v$ is received by $u$ for the first time}


            \State $u.G$ is updated as a subgraph of $u.G$ induced by $u.V_G$ $-$ $\{v\}$

            \State $u.G_{r+}$ is updated as a subgraph of $u.G_{r+}$ induced by $u.V_{G_{r+}}$ $-$ $X$, where $X$ denotes the set of nodes $v_{z,w}^{j} \in u.V_{G_{r+}}$ with $z$ $=$ $v$ or $w$ $=$ $v$

         \EndIf

         \If {$\phi(u)$ $=$ $0$ and $hop(u)$ $\geq$ $hop(v)$ for all $v \in u.V_G$}

            \State $u.finished \gets true$

            \State locally broadcast a MSG$\_$FINISHED message with $u$ to all sensors in $u.V_G$

         \ElsIf {$u.t \le v.t$ for all $v \in u.V_G$}

            \State SCHEDULE(u)

         \EndIf

    \EndWhile

    \State return $u.S$
\end{algorithmic}
\end{algorithm*}

Take Fig. \ref{Fig:net_model_org}, for example. We assume that the
aggregation ratio $\alpha$ = $3$ and the total number of channels
$\varsigma$ $=$ $2$. When each sensor $u$ in the network executes the DCAS algorithm,
$u.t$ and $u.S$ are set to $1$ and $\emptyset$, respectively.
In addition, all nodes in $u.V_{G_{r+}}$ are marked as white.
For sensor $f$ in the network, because $f.t = 1$ that is less than or equal to $v.t$ for all $v \in f.V_G$,
that is, $f.t$ $\leq$ $a.t$ $=$ $b.t$ $=$ $d.t$ $=$ $e.t$ $=$ $s.t$, $f$ can execute Procedure SCHEDULE.
When $f$ executes Procedure SCHEDULE, $V$ $=$ $\{v_{a,s}^1$, $v_{e,a}^1$, $v_{d,a}^1$, $v_{f,d}^1$, $v_{a,s}^2$, $v_{e,a}^2$, $v_{d,a}^2$, $v_{f,d}^2\}$ because all nodes in $f.V_{G_{r+}}$ of Fig. \ref{Fig:sub_ERC}
are white. By Definition \ref{def:precedence}, we have that $v_{f,d}^1$ is the node having highest precedence in $V$.
We also have that $f$ will locally broadcast a MSG$\_$DECISION message with $S_{1}$ $=$ ($f \rightarrow d$, $\varphi$, $ch_1$)
to all sensors in $f.V_G$ of Fig. \ref{Fig:sub_forwarding_graph}, where $\varphi$ $=$ $min(\phi(f), \alpha)$ $=$ $min(1, 3)$ $=$ $1$.
In addition, $\phi(f)$, $\phi(d)$, $f.t$, and $d.t$ are updated to $0$, $3$, $2$, and $2$, respectively.
In the same way, sensor $c$ will locally broadcast a MSG$\_$DECISION message with $S_{1}$ $=$ ($c \rightarrow b$, $3$, $ch_1$)
to all sensors in $c.V_G$, and $c.t$ is incremented by 1. When sensor $a$ receives MSG$\_$DECISION messages from $f$ and $c$, all nodes in $a.V_{G_{r+}}$ except for nodes $v_{a,s}^2$ and $v_{e,a}^2$ become black. Because $v_{a,s}^2$ has higher precedence than $v_{e,a}^2$,
sensor $a$ will locally broadcast a MSG$\_$DECISION message with $S_{1}$ $=$ ($a \rightarrow s$, $3$, $ch_2$)
to all sensors in $a.V_G$, and $a.t$ is incremented by 1. When sensor $e$ receives MSG$\_$DECISION messages from $a$, $c$, and $f$, all nodes in $e.V_{G_{r+}}$ become black, and then, $e$ will locally broadcast a MSG$\_$SKIP message to all sensors in $e.V_G$.

When $f.t$ is changed to $2$, because $\phi(f)$ $=$ $0$ and $hop(f)$ $\geq$ $hop(v)$ for all $v \in f.V_G$,
sensor $f$ locally broadcasts a MSG$\_$FINISHED message with $f$ to all sensors in $f.V_G$ and finishes its scheduling process.
In a similar way, after scheduling a number of time slots, sensors will gradually finish the scheduling processes from the outer network to the inner network. When all the neighboring sensors of the sink finish the scheduling processes, the data scheduling
for the network is completed.

\section{Analysis of the DCAS}\label{section:analysis}
In this section, we first prove that the data transmissions scheduled by the DCAS are collision-free
in Theorem \ref{thm:correctness}.
In addition, in the DCAS, a sensor $u$ can make a schedule of transmitting
data from itself to another sensor $y$ with channel $i$ if $v_{u,y}^{i}$ has higher precedence than other possible data transmissions $v_{z,w}^{j}$ in $u.V_{G_{r+}}$;
otherwise, $u$ has to skip the current time slot or wait until other sensors make decisions.
Theorem \ref{thm:no_circular wait} shows that no circular wait will occur in sensors by the DCAS.

\begin{thm}\label{thm:correctness}
The data transmissions scheduled by the DCAS are
collision-free in each time slot.
\end{thm}

\begin{proof}
By Definition \ref{def:collision}, it suffices to show that no cases in Definition \ref{def:collision}
are occurred by the DCAS. In the DCAS, when sensor $u$ is scheduled to be a transmitter at time slot $u.t$,
$u.t$ is changed to the next time slot immediately.
We have that at most one data transmission sent from $u$ is scheduled at time slot $u.t$, which implies that the case C1 in Definition \ref{def:collision} will not occur in the DCAS. Therefore, it suffices to show
that the cases C2, C3, and C4 in Definition \ref{def:collision} will not occur in the DCAS.
Because the proofs for showing that the cases C2 and C3 are not occurred in the DCAS are similar to
that for the case C4, the proofs for C2 and C3 are omitted here.

Here, we prove that the case C4 will not occur in the DCAS by contradiction.
We assume that there exists at least one case C4 in the DCAS, that is,
there exist one data transmission from sensor $x$ to sensor $u$
and one data transmission from sensor $w$ to sensor $y$ scheduled with channel $i$ at time slot $t$ in the DCAS,
where $u$ and $w$ are neighboring sensors.
This implies that $x$ (or $w$) can make a schedule of transmitting data from $x$ (or $w$)
to its neighboring sensor $u$ (or $y$) with channel $i$ for time slot $t$.
This also implies that $v^i_{x,u}$ (or $v^i_{w,y}$) has higher precedence than other
possible data transmissions $v^j_{p,q}$ in $x.V_{G_{r+}}$ (or $w.V_{G_{r+}}$).
Because $x$ and $w$ are $u$'s neighboring sensors, $x$ is within $2$-hop distance from $w$, and
we have that $v^i_{x,u}$ and $v^i_{w,y}$ are both in $w.V_{G_{r+}}$ and $x.V_{G_{r+}}$ before any schedules made by $x$ and $w$.
Because the data transmissions from $x$ to $u$ and from $w$ to $y$ are finally scheduled,
it implies that $v^i_{x,u}$ has higher precedence than $v^i_{w,y}$, and $v^i_{w,y}$ has higher precedence
than $v^i_{x,u}$, which constitutes a contradiction by Definition \ref{def:precedence}, and thus completes the proof.
\end{proof}

\begin{thm}\label{thm:no_circular wait}
In the DCAS, no sensors are in a circular wait for making schedules of data transmissions.
\end{thm}

\begin{proof}
Assume that there exist sensors $u_0$, $u_1$, $u_2$, $\ldots$, $u_n$ in a circular wait for making schedules of data transmissions.
It implies that there exist some data transmission $v^{{c_1}}_{u_1,{p_1}}$ having higher precedence than $v^{{c'_0}}_{u_0,{p'_0}}$ for any ${c'_0}$ and ${p'_0}$,
some data transmission $v^{{c_2}}_{u_2,{p_2}}$ having higher precedence than $v^{{c'_1}}_{u_1,{p'_1}}$ for any ${c'_1}$ and ${p'_1}$,
some data transmission $v^{{c_j}}_{u_{j},{p_j}}$ having higher precedence than $v^{{c'_{j-1}}}_{u_{j-1},{p'_{j-1}}}$ for any ${c'_{j-1}}$ and ${p'_{j-1}}$,
where $1 \leq j \leq n$.
In addition, due to the circular wait, there exists some data transmission $v^{{c_0}}_{u_0,{p_0}}$ having higher precedence than $v^{{c'_n}}_{u_n,{p'_n}}$ for any ${c'_n}$ and ${p'_n}$. Let $prec(v^{{c_i}}_{u_i,{p_i}})$ denote the precedence of the data transmission $v^{{c_i}}_{u_i,{p_i}}$.
We therefore have that $prec(v^{{c_0}}_{u_0,{p_0}})$ $<$ $prec(v^{{c_1}}_{u_1,{p_1}})$ $<$ $\ldots$ $<$
$prec(v^{{c_n}}_{u_n,{p_n}})$ $<$ $prec(v^{{c_0}}_{u_0,{p_0}})$, which constitutes a contradiction.
This thus completes the proof.
\end{proof}

\section{Performance Evaluation}\label{section:Simulation}
In this section, simulations were used to evaluate the performance
of the DCAS. In the simulations, $50$-$350$ sensors, including a sink,
were randomly deployed in a $L \times L$ square area,
where the transmission range of each sensor was set to $5$ and $L \in {{\mathbb{Z}}^{+}}$.
In addition, the number of units of raw data generated by each sensor was randomly selected from the interval $[1,\beta]$, where $\beta \in \mathbb{Z}^+$.
To demonstrate the performance of the DCAS, the RMSA \cite{Bagaa2014293}, which is used for multi-channel scheduling in WSNs, is considered to be compared.
In the RMSA, the aggregated data is allowed to be forwarded to multiple selected nodes for collection reliability.
In the simulation, the aggregated data is only forwarded to one node selected by the RMSA.
In addition, because the RMSA is only used for the WSNs with $\alpha = \infty$,
to apply to the multi-channel WSNs with different $\alpha$,
the schedule generated by the RMSA is periodically used for forwarding aggregated data, until all data is collected by the sink.
If one time slot has no aggregated data required to be forwarded, the time slot is removed from the schedule.
In the following subsections, the simulation results were obtained by averaging $100$ data.

\subsection{Impact of the Number of Channels}\label{section:change_channel}
Here, we show how the number of channels affects the
performance of our proposed algorithm.
Fig. \ref{Fig:compare_channel_a}, Fig. \ref{Fig:compare_channel_b}, and Fig. \ref{Fig:compare_channel_c}
illustrate the results in terms of the total number of the time slots required
in WSNs with $1$, $2$, and $4$ channels, respectively, when $\beta$ $=$ $3$, the network size is $30 \times 30$, the number of sensors is $200$, and the data aggregation ratio $\alpha$ ranges from $1$ to $128$.
In Fig. \ref{Fig:compare_channel_a}, it is clear that
the DCAS outperforms the RMSA in WSNs with different $\alpha$ value because
the capability of aggregating data and the utilization of multiple channels are both considered in the DCAS.
In addition, the higher the $\alpha$ value, the significantly lower the total number of the time slots required by the RMSA or the DCAS. This is because more raw data can be aggregated into one packet with
an increasing $\alpha$ value.
Similar results are also shown in Fig. \ref{Fig:compare_channel_b} (or Fig. \ref{Fig:compare_channel_c}).
Moreover, by comparing the results in Fig. \ref{Fig:compare_channel_a} and
Fig. \ref{Fig:compare_channel_b} (or Fig. \ref{Fig:compare_channel_b} and Fig. \ref{Fig:compare_channel_c}),
the RMSA (or the DCAS) requires fewer time slots when the number of
channels increases. This stems from the fact that when more channels can be used,
more data transmissions can be scheduled into one time slot, and thus, the total number of the required
time slots is decreased.

\begin{figure}
\center
\subfigure[]{\includegraphics[width=7.5cm]{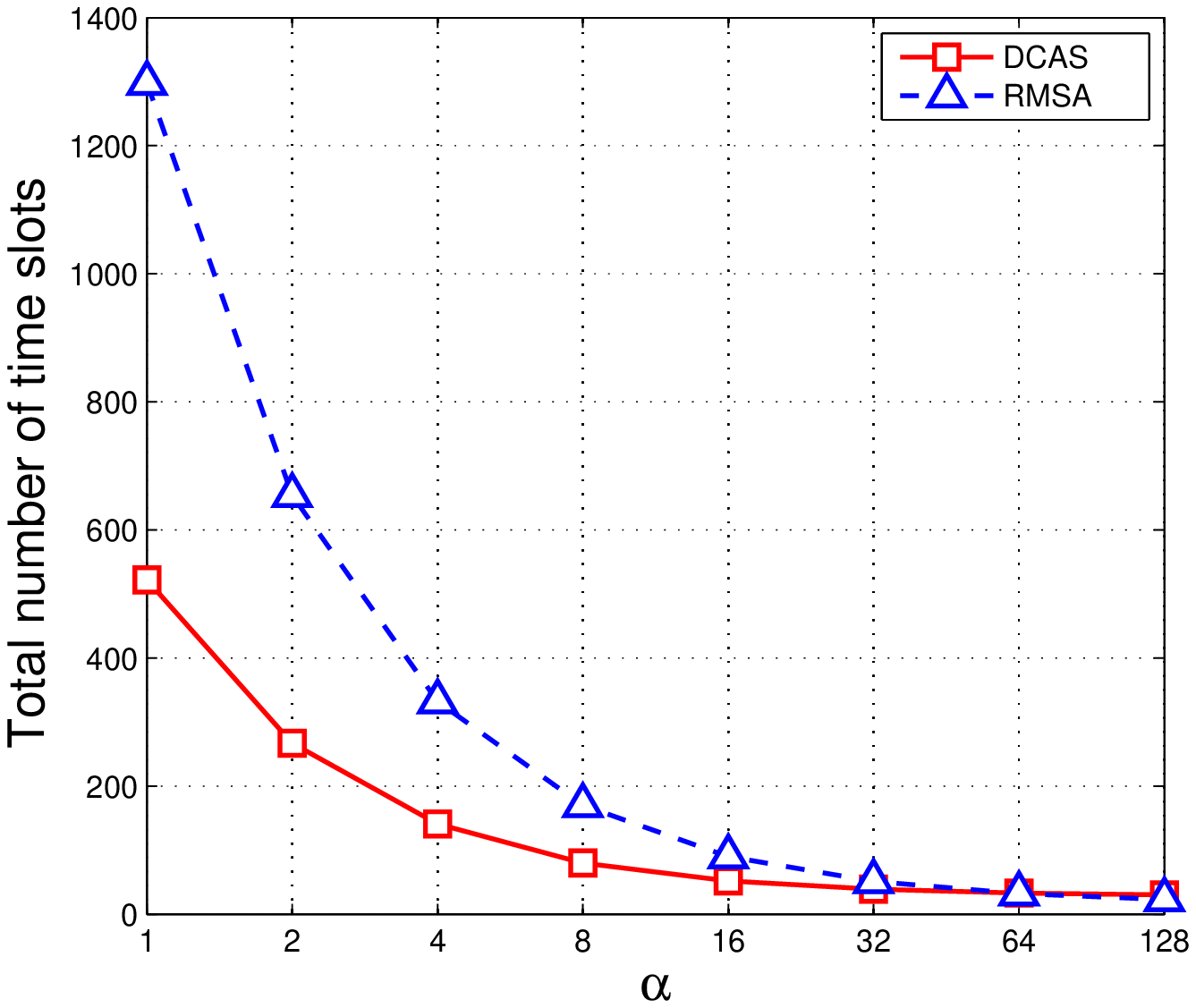} \label{Fig:compare_channel_a}}
\subfigure[]{\includegraphics[width=7.5cm]{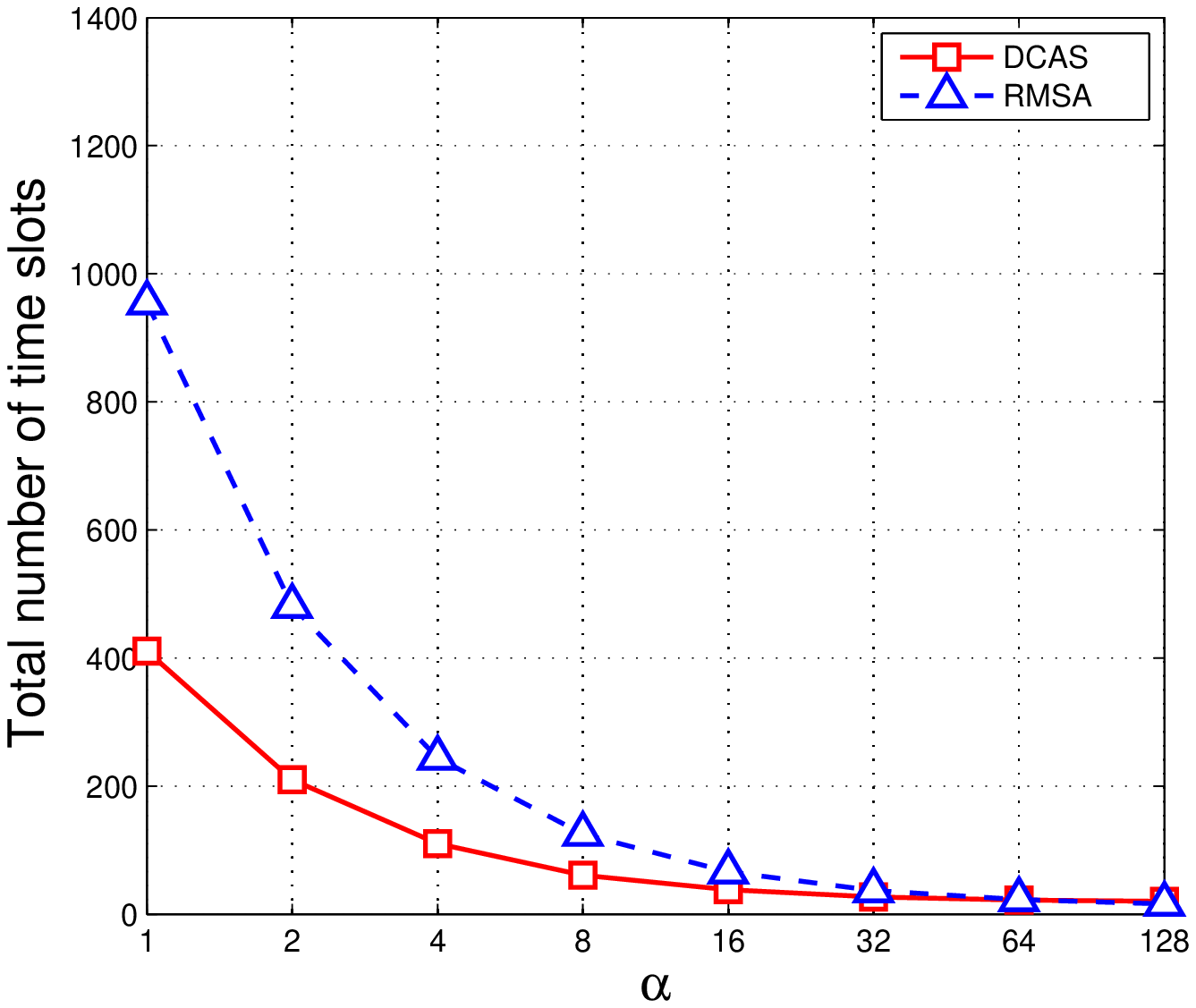} \label{Fig:compare_channel_b}}
\subfigure[]{\includegraphics[width=7.5cm]{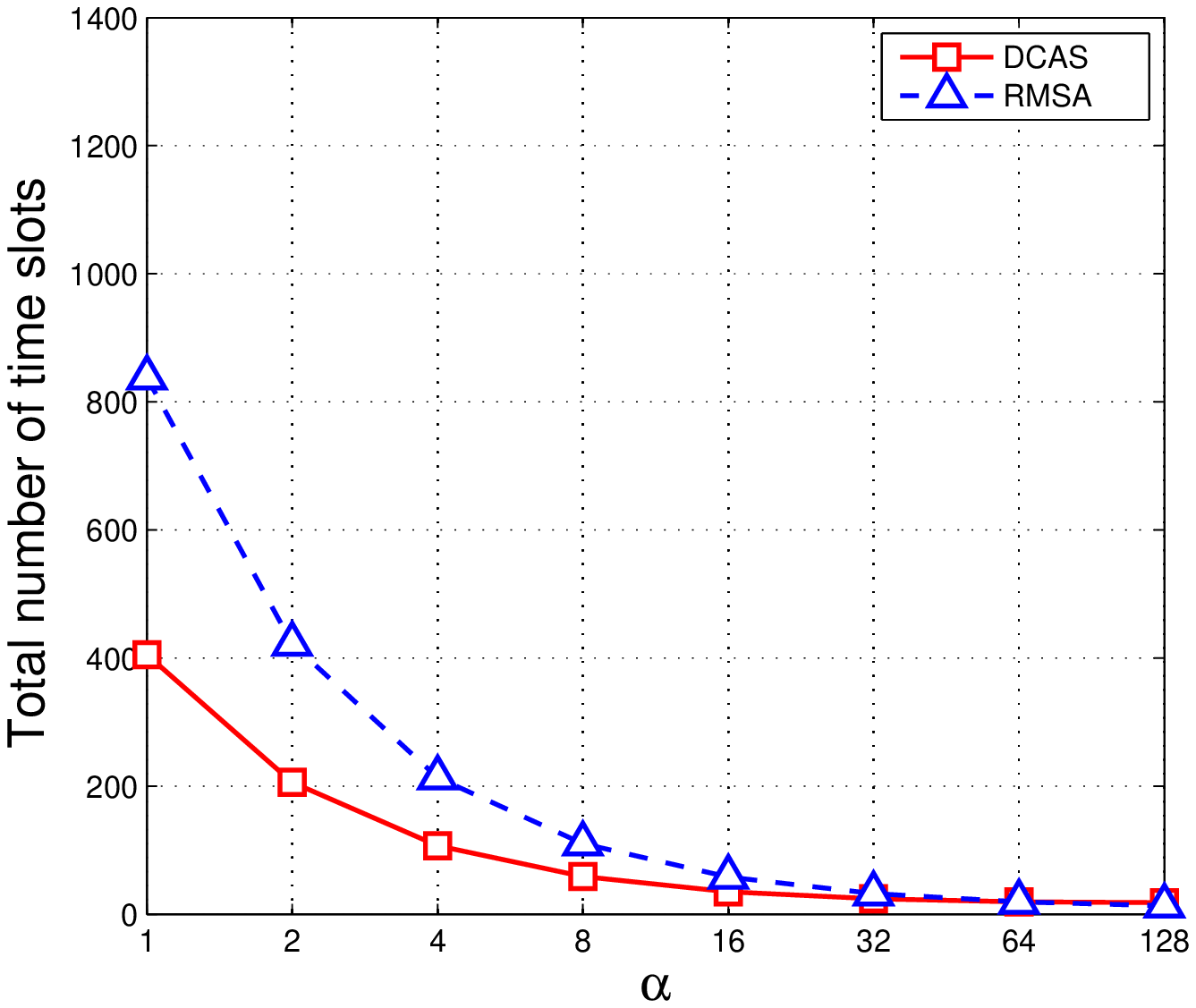} \label{Fig:compare_channel_c}}
\caption{The total number of the time slots required by the RMSA
and the DCAS in WSNs with $\beta$ $=$ $3$, the network size equal to $30 \times 30$, the number of sensors equal to $200$, and $\alpha$ ranging
from $1$ to $128$. The numbers of channels are $1$ in (a), $2$ in (b), and $4$ in (c), respectively.} \label{Fig:compare_channel}
\end{figure}

\subsection{Impact of the Number of Sensors}\label{section:change_sensor}
We show how the number of sensors affects the
performance of our proposed algorithm in this section.
Fig. \ref{Fig:compare_sensor_a}, Fig. \ref{Fig:compare_sensor_b}, and Fig. \ref{Fig:compare_sensor_c} show the simulation results concerning the total number of the time slots required in WSNs with $50$, $200$, and $350$ sensors, respectively,
when $\beta$ $=$ $3$, the network size is $30 \times 30$, the number of channels is $2$, and $\alpha$ ranges from $1$ to $128$.
In these figures, it is clear that the DCAS outperforms the RMSA in WSNs with the number of sensors equal to $50$, $200$, or $350$.
It is also clear that the proportion of the results are similar in
Fig. \ref{Fig:compare_sensor_a}, Fig. \ref{Fig:compare_sensor_b}, and Fig. \ref{Fig:compare_sensor_c}.
When $\alpha$ value increases, the total number of the time slots required by the RMSA or the DCAS is decreased
because more raw data are allowed to be aggregated into one packet.
Fig. \ref{Fig:compare_sensor_a},
Fig. \ref{Fig:compare_sensor_b}, and Fig. \ref{Fig:compare_sensor_c}
also show that the RMSA (or the DCAS) requires more time slots when the number of
sensors increases. This is because more raw data are generated and required to be forwarded to the sink,
and therefore, more time slots are required by the RMSA (or the DCAS).

\begin{figure}
\center
\subfigure[]{\includegraphics[width=7.5cm]{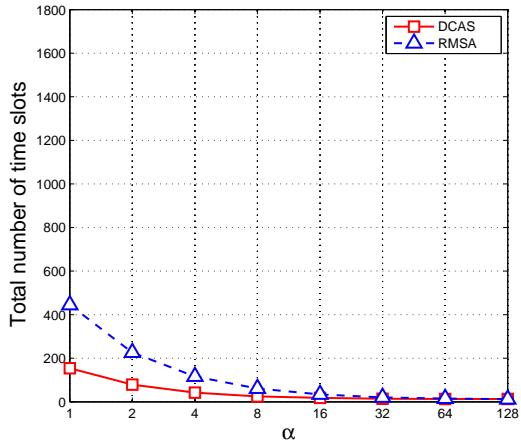} \label{Fig:compare_sensor_a}}
\subfigure[]{\includegraphics[width=7.5cm]{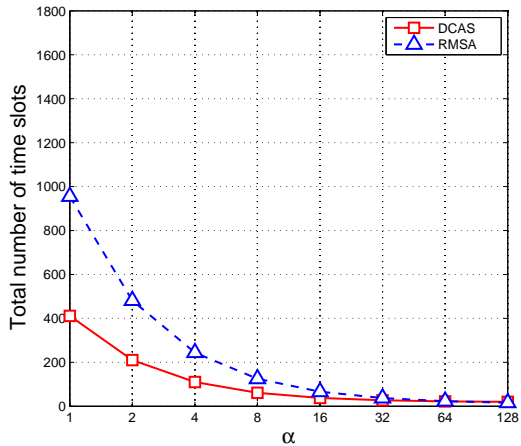} \label{Fig:compare_sensor_b}}
\subfigure[]{\includegraphics[width=7.5cm]{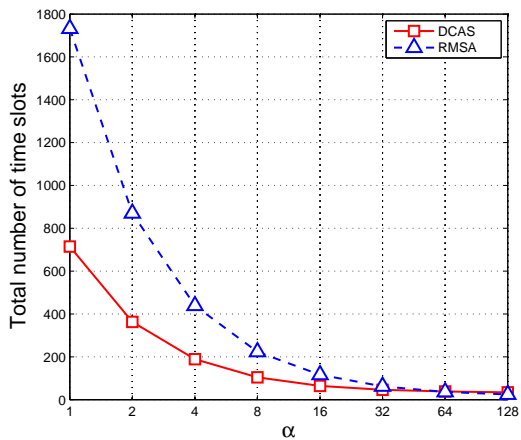} \label{Fig:compare_sensor_c}}
\caption{The total number of the time slots required by the RMSA
and the DCAS in WSNs with $\beta$ $=$ $3$, the network size equal to $30 \times 30$, the number of channels equal to $2$, and $\alpha$ ranging
from $1$ to $128$. The numbers of sensors are $50$ in (a), $200$ in (b), and $350$ in (c), respectively.
} \label{Fig:compare_sensor}
\end{figure}

\subsection{Impact of the Network Size}\label{section:change_net_size}
In this section, we show how the network size affects the
performance of the DCAS. Fig. \ref{Fig:compare_size_a},
Fig. \ref{Fig:compare_size_b}, and Fig. \ref{Fig:compare_size_c}
illustrate the results in terms of the number of the time slots
required in WSNs with network sizes equal to $20 \times 20$, $30 \times 30$, and $40 \times 40$, respectively,
when $\beta$ $=$ $3$, the number of sensors is $200$, the number of channels is $2$, and $\alpha$ ranges from $1$ to $128$.
As observed in previous simulation results, the DCAS requires a significantly lower number of time slots than the RMSA in most of cases in
Fig. \ref{Fig:compare_size_a}, Fig. \ref{Fig:compare_size_b}, and Fig. \ref{Fig:compare_size_c}.
In addition, the more the $\alpha$ value, the lower the number of the time slots required by the RMSA (or the DCAS), which is also observed
in previous simulation results.
Note that in Fig. \ref{Fig:compare_size_a}, Fig. \ref{Fig:compare_size_b}, and Fig. \ref{Fig:compare_size_c},
when the network size increases, the DCAS requires an increasing number of time slots.
Also note that in Fig. \ref{Fig:compare_size_a} and Fig. \ref{Fig:compare_size_c} (or in Fig. \ref{Fig:compare_size_b} and Fig. \ref{Fig:compare_size_c}),
the RMSA has similar results.
This is because when $200$ sensors are randomly distributed in a larger sensing field,
the minimum hop count between sensors
and the sink has a probability to be increased, and thus, more time slots are required by the RMSA (or the DCAS).
It is noted that in Fig. \ref{Fig:compare_size_a} and Fig. \ref{Fig:compare_size_b}, the RMSA has a decreasing number of time slots when the network size increases.
This is because when the network size increases, the number of neighboring sensors has a probability to be decreased,
and a scheduled data transmission has a lower probability to affect sensors for other data transmissions, which dominates the performance of the RMSA in Fig. \ref{Fig:compare_size_b}.
And therefore, there is a higher probability to schedule more data transmissions in one time slot in the RMSA.

\begin{figure}
\center
\subfigure[]{\includegraphics[width=7.5cm]{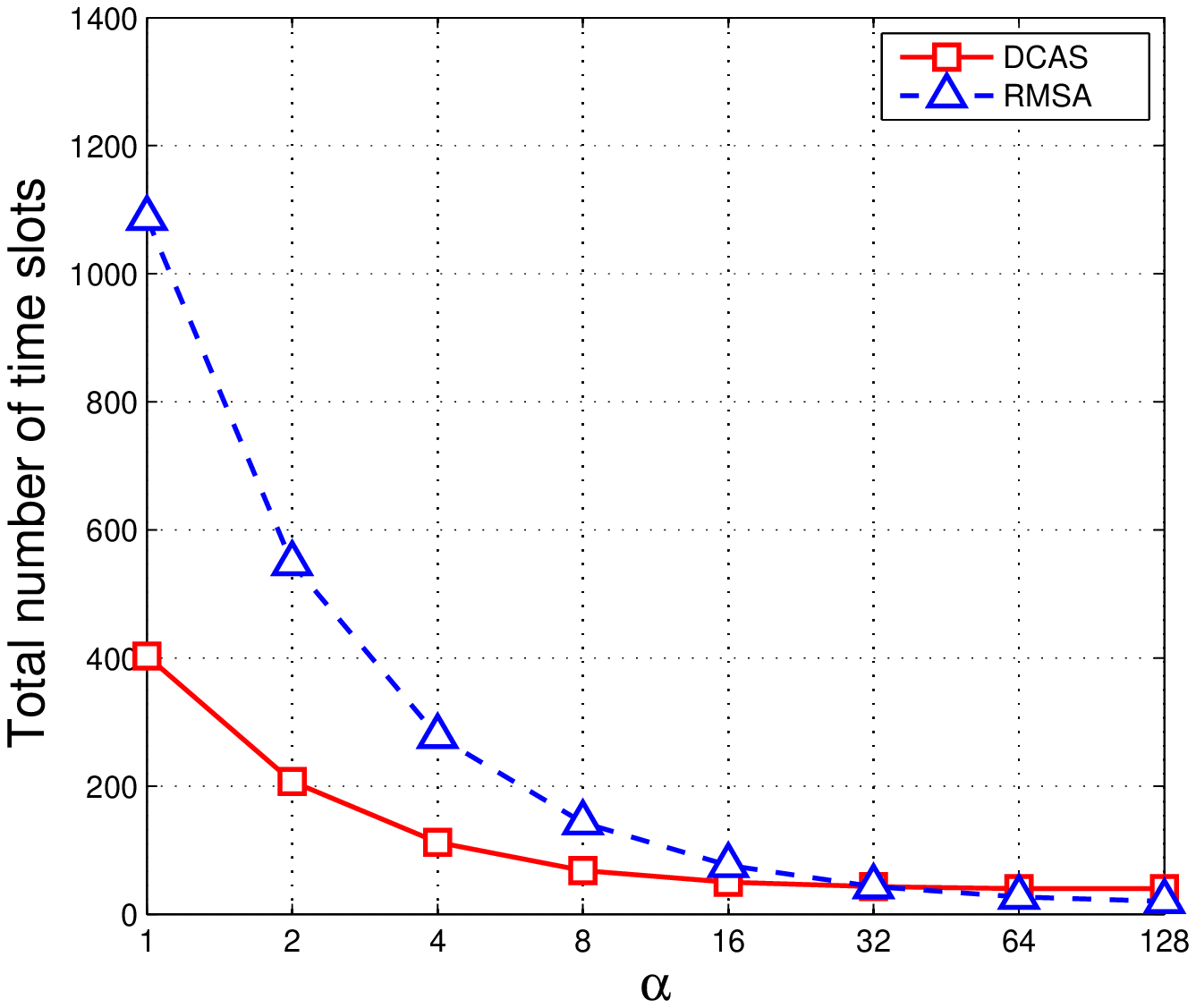} \label{Fig:compare_size_a}}
\subfigure[]{\includegraphics[width=7.5cm]{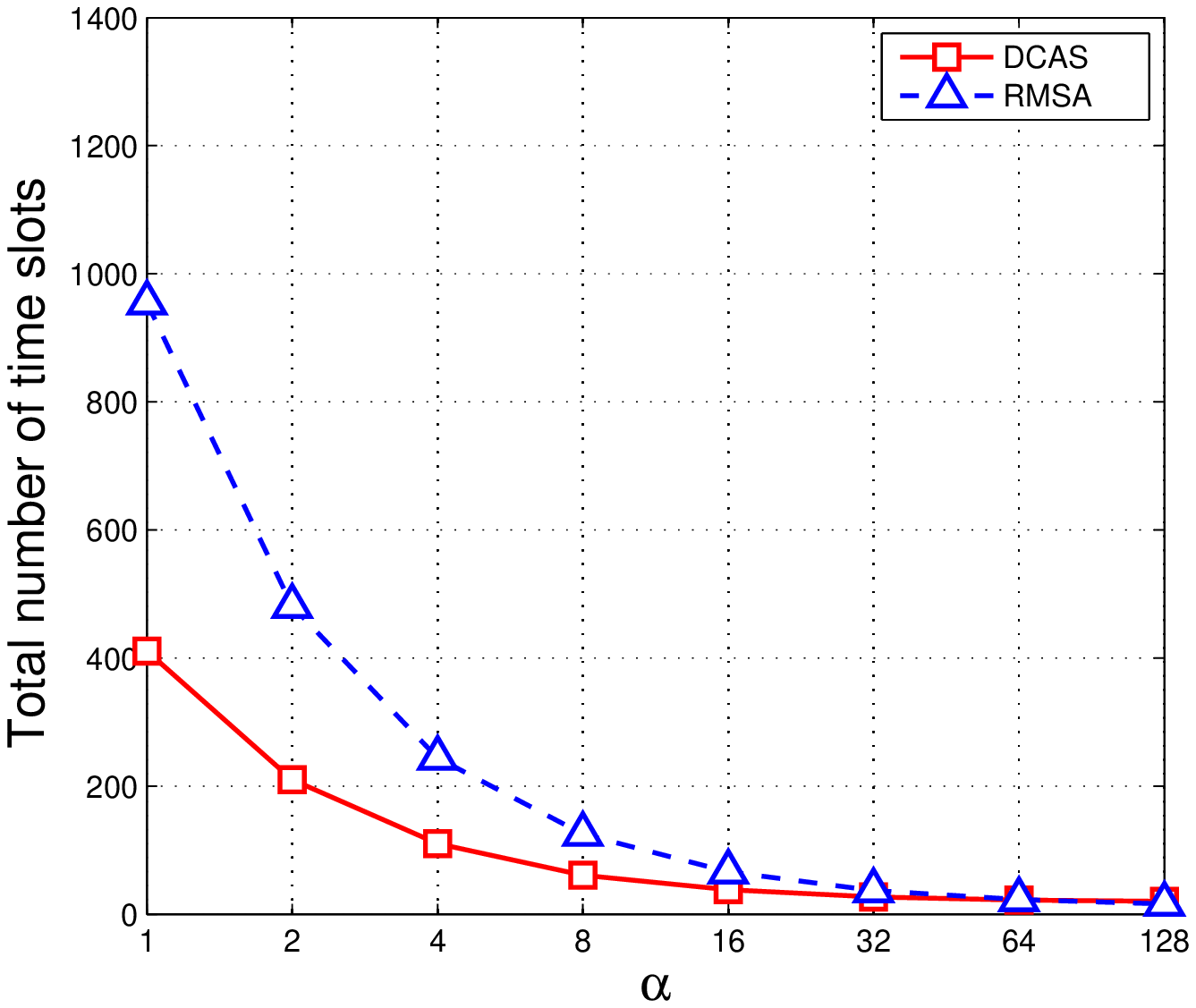} \label{Fig:compare_size_b}}
\subfigure[]{\includegraphics[width=7.5cm]{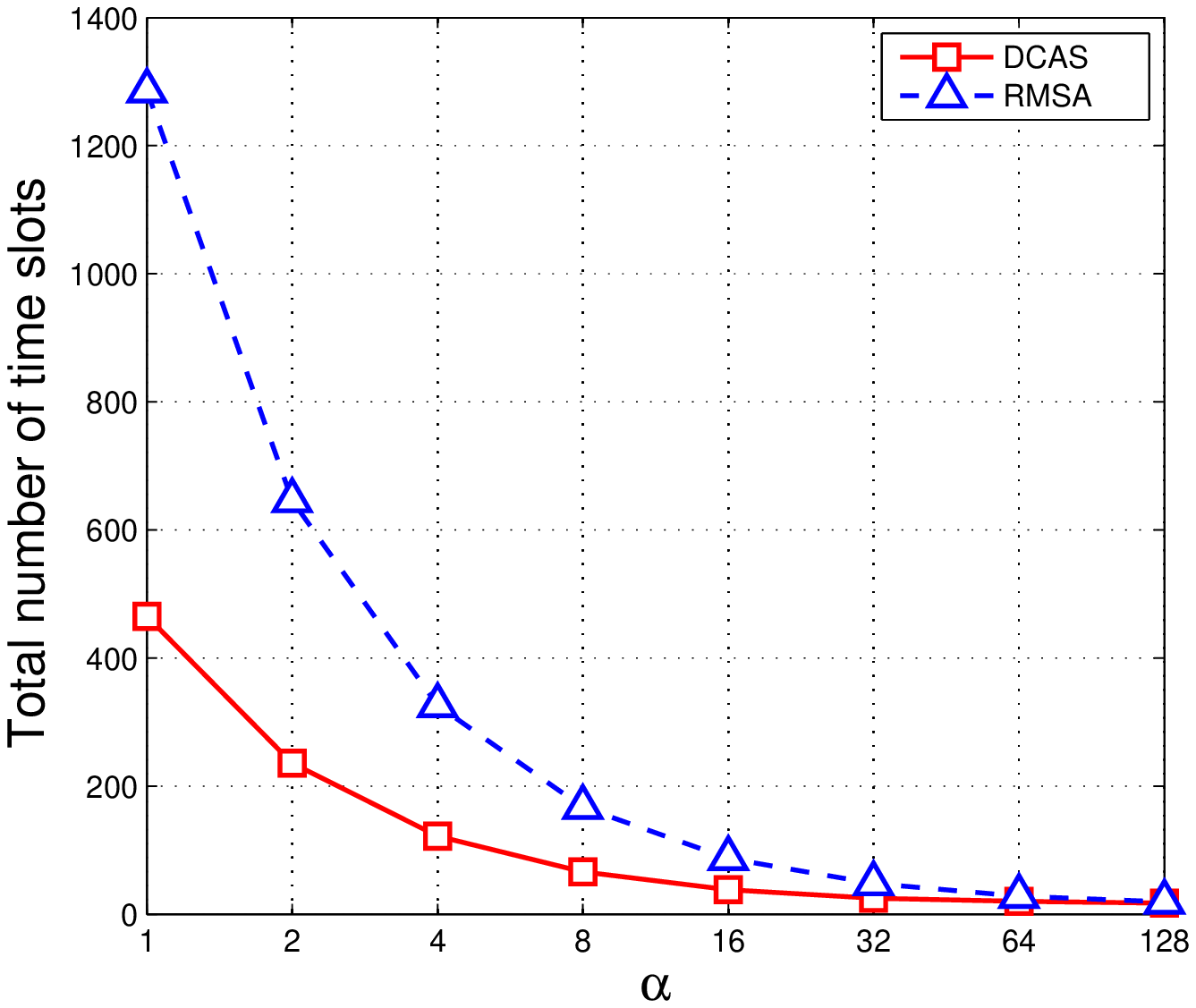} \label{Fig:compare_size_c}}
\caption{The total number of the time slots required by the RMSA
and the DCAS in WSNs with $\beta$ $=$ $3$, the number of sensors equal to $200$, the number of channels equal to $2$, and $\alpha$ ranging
from $1$ to $128$. The network sizes are $20 \times 20$ in (a), $30 \times 30$ in (b), and $40 \times 40$ in (c), respectively.
} \label{Fig:compare_size}
\end{figure}

\subsection{Impact of the Number of Units of Raw Data Generated by Sensors}\label{section:change_data}
In this section, we show how the number of units of raw data generated by sensors affects the
performance of our proposed algorithm.
Fig. \ref{Fig:compare_data_a},
Fig. \ref{Fig:compare_data_b} and Fig. \ref{Fig:compare_data_c} show the simulation results concerning the number of the time slots
required in WSNs with the values of $\beta$ equal to $1$, $3$, and $5$, respectively,
when the network size is $30 \times 30$, the number of sensors is $200$, the number of channels is $2$, and $\alpha$ ranges from $1$ to $128$.
In Fig. \ref{Fig:compare_data_a}, Fig. \ref{Fig:compare_data_b}, and Fig. \ref{Fig:compare_data_c},
it is noted that the DCAS provides a better performance than the RMSA, and that the DCAS (or the RMSA) has a lower
number of the time slots with an increasing $\alpha$ value, as observed in the previous simulation results.
In addition, when the value of $\beta$ increases from $1$, through $3$, to $5$,
the number of the time slots required by the RMSA (or the DCAS) is increased.
This is because more raw data are generated by sensors, more packets have to be forwarded to the sink.
And therefore, more time slots are required by the RMSA (or the DCAS).

\begin{figure}
\center
\subfigure[]{\includegraphics[width=7.5cm]{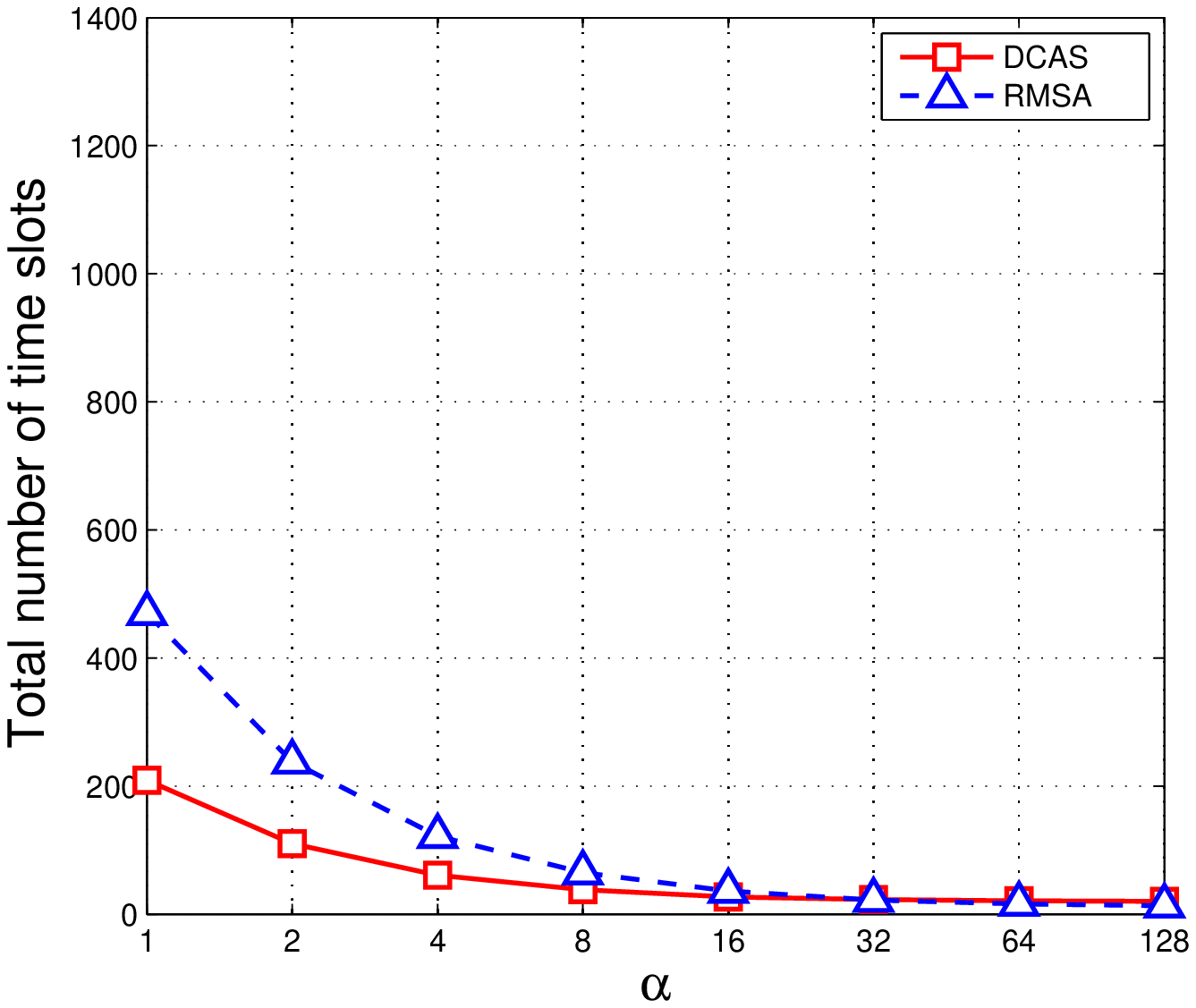} \label{Fig:compare_data_a}}
\subfigure[]{\includegraphics[width=7.5cm]{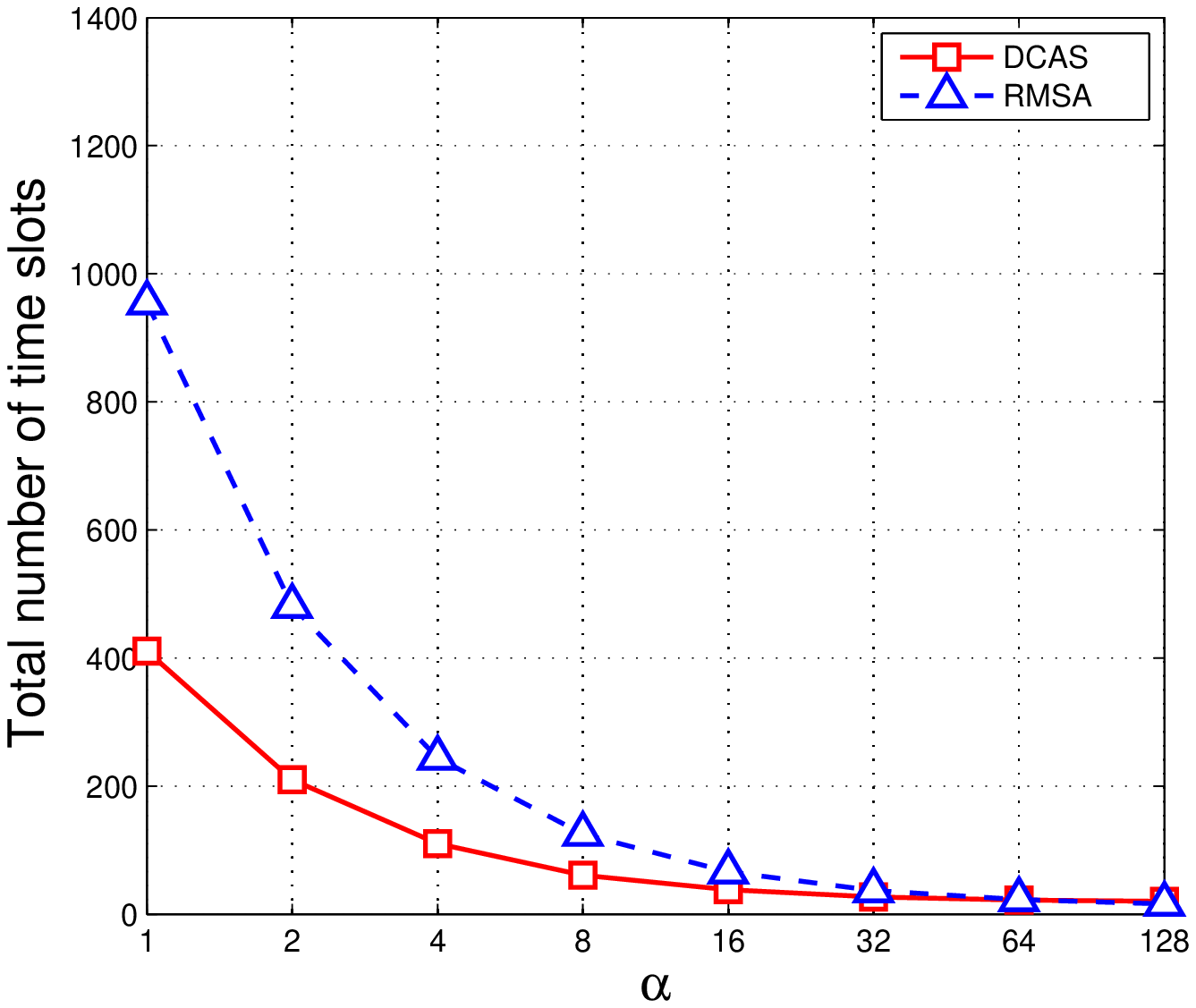} \label{Fig:compare_data_b}}
\subfigure[]{\includegraphics[width=7.5cm]{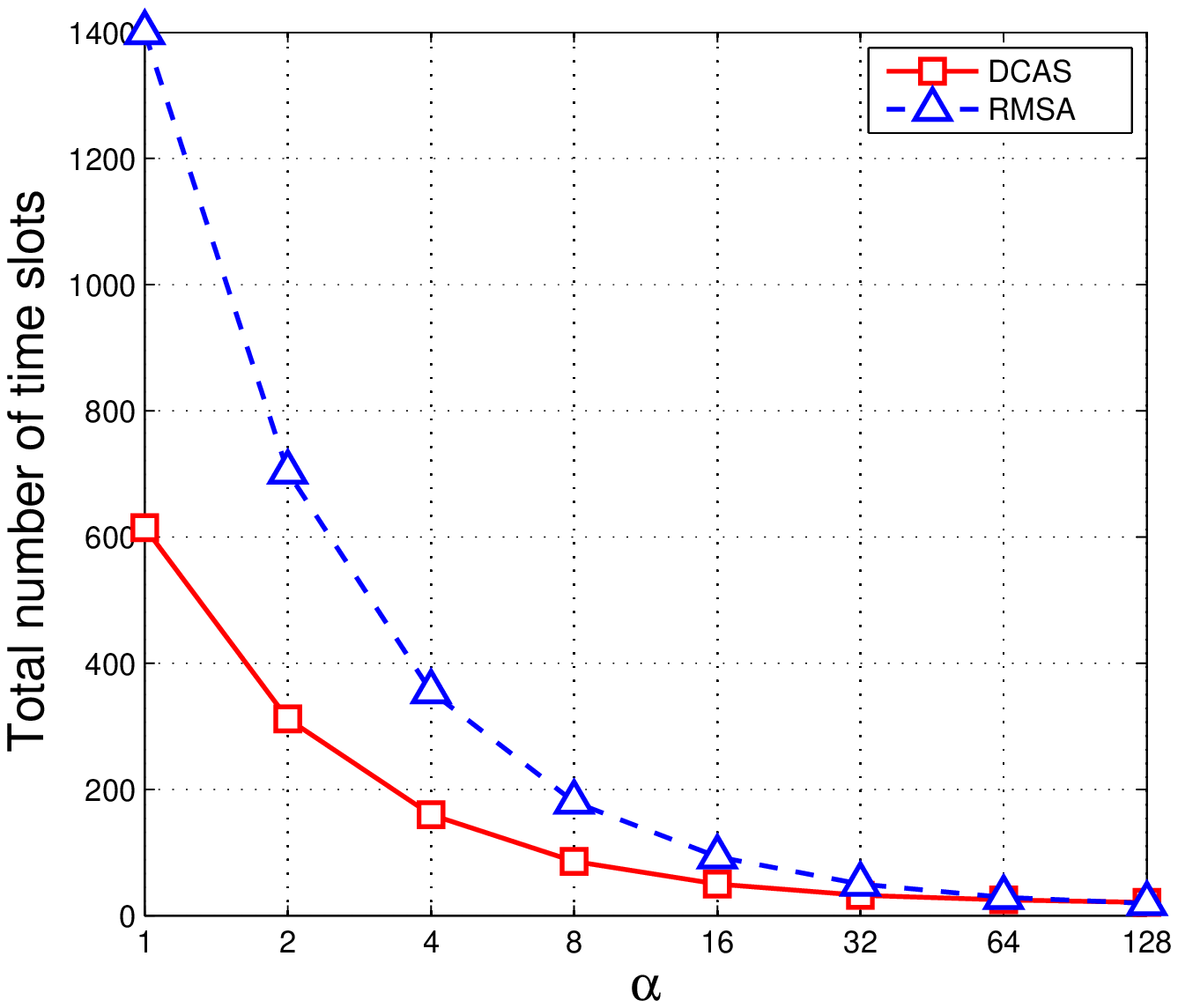} \label{Fig:compare_data_c}}
\caption{The total number of the time slots required by the RMSA
and the DCAS in WSNs with the network size equal to $30 \times 30$, the number of sensors equal to $200$, the number of channels equal to $2$, and $\alpha$ ranging
from $1$ to $128$. The values of $\beta$ are $1$ in (a), $3$ in (b), and $5$ in (c), respectively.}
\end{figure}

\section{Conclusion}\label{section:Conclusion}
Guaranteeing minimum latency of data collection in distributed WSNs as well as eliminating
the data collisions is a challenging undertaking. In this paper, we investigated the problem
of Minimum-Latency Collision-Avoidance Multiple-Data-Aggregation Scheduling with
Multi-Channel (MLCAMDAS-MC). To avoid data collisions, we construct a extended relative collision graph $G_{r+}$ to represent the data collisions
of transmissions. Subsequently, based on the obtained $G_{r+}$, we propose
a distributed collision-avoidance scheduling (DCAS) algorithm with fairness assumption for distributed WSNs.
Theoretical analyses of DCAS shows that the data collisions are completely eliminated
in the network. Extensive simulation results demonstrate that DCAS better performance
with a lower number of required time slots than the most recently published distributed algorithm.


\normalem

\bibliographystyle{IEEEtran}
\bibliography{my}

\end{document}